\documentclass[11pt]{article}

\usepackage{color}

\usepackage[utf8]{inputenc}

\usepackage{float}
\usepackage{tikz}
\usepackage{csquotes}
\usetikzlibrary{positioning,arrows}
\tikzset{
m/.style={circle,draw,fill=gray!20,minimum size=5},outer sep=2pt}

\usepackage{amsmath,amsfonts,amssymb,amsthm, mathrsfs}

\usepackage{dcolumn}

\newcolumntype{d}[1]{D{.}{.}{#1} } 

\usepackage{verbatim}
\usepackage{bbm}

\usepackage{graphicx,wrapfig, tikz, cancel, hyperref, multirow,rotating}
\usepackage{caption, subcaption}
\usepackage{natbib}
\usepackage{enumerate}
\usepackage[overload]{empheq}
\usepackage{pdflscape} 
\usepackage{hyperref}
\hypersetup{colorlinks,            citecolor=blue,            filecolor=blue,            linkcolor=blue,            urlcolor=blue,            }

\theoremstyle{plain}
\newtheorem{prop}{Proposition}
\newtheorem{lemma}{Lemma}

\theoremstyle{definition}
\newtheorem{assumption}{Assumption}

\newcommand{\dd}{\mathrm{d}}

\newcommand{\EE}{\mathbb{E}}

\newcommand{\vs}{\vspace{5 mm}}

\usepackage{authblk}

\begin{document}

\title{Epidemic dynamics with homophily, vaccination choices, and pseudoscience attitudes \thanks{%
We thank seminar participants at the University of Antwerp, Universit\`a di Bologna, Oxford University, at the 2nd EAYE Conference at Paris School of Economics, Bocconi University, and Politecnico di Torino. We also would like to thank Daniele Cassese, Andrea Galeotti, Melika Liporace, Matthew Jackson, Alessia Melegaro, Luca Merlino, Dunia Lopez--Pintado, Davide Taglialatela, and Fernando Vega--Redondo for useful comments and suggestions. Fabrizio Panebianco and Paolo Pin gratefully acknowledges
funding from the Spanish Ministry of Economia y competitividad project ECO2107-87245-R, and Italian Ministry of Education Progetti di Rilevante Interesse
Nazionale (PRIN) grants 2015592CTH and 2017ELHNNJ, respectively.}}

\author[a,e]{Matteo Bizzarri}
\author[b]{Fabrizio Panebianco}
\author[c,d]{Paolo Pin}
\affil[a]{Universit\`a  Bocconi, Milan, Italy}
\affil[b]{Department of Economics and Finance, Universit\`a Cattolica, Milan, Italy}
\affil[c]{Department of Economics and Statistics, Universit\`a di Siena, Italy}
\affil[d]{BIDSA, Universit\`a  Bocconi, Milan, Italy}
\affil[e]{CSEF, Napoli}

\date{June 2021 }

\maketitle

%



\vspace{-0.5cm}
\begin{abstract}

We interpret attitudes towards science and pseudosciences as cultural traits that diffuse in society through communication efforts exerted by agents. We present a tractable model that allows us to study the interaction among the diffusion of an epidemic, vaccination choices, and the dynamics of cultural traits. We apply it to study the impact of homophily between  \emph{pro-vaxxers} and \emph{anti-vaxxers} on the total number of cases (the \emph{cumulative infection}).
We show that, during the outbreak of a disease, homophily has the direct effect of decreasing the speed of recovery. Hence, it may increase the number of cases and make the disease endemic. 
The dynamics of the shares of the two cultural traits in the population is crucial in determining the sign of the total effect on the cumulative infection: more homophily is beneficial if agents are not too flexible in changing their cultural trait, is detrimental otherwise.
\end{abstract}

	\textbf{D62}	Externalities --
	\textbf{D85}	Network Formation and Analysis: Theory --
		\textbf{I12}	Health Behavior --
	\textbf{I18}	Government Policy, Regulation, Public Health

\bigskip


	\section{Introduction}
	
We model an economy facing the possible outbreak of a disease, for which a vaccine with temporary efficacy is available.
	This mimics what happens every year for seasonal flu, but it could also be the case in the near future for Covid19.\footnote{%
	At present, we know that the virus of Covid19 mutates very rapidly \citep{korber2020spike,pachetti2020emerging} and that it seems to be seasonal \citep{carleton2020causal}. Scientists and politicians are considering the possibility that, for the next year, it could become similar to a seasonal flu that deserves a new vaccine every year: for example, see  \href{https://www.healthline.com/health-news/covid19-could-become-seasonal-just-like-the-flu}{this report from April 2020}.
	There is also another reason for which vaccination against Covid19 may not be permanent: more recent studies like \cite{seow2020longitudinal} have shown that Covid19 antibodies fall rapidly in our body so that it could be the case that people will need to vaccinate regularly (e.g.~once every year) against the virus.}
	
	Even before Covid19, vaccination has been almost unanimously considered the most effective public health intervention by the scientific community (see, e.g., \citealp{larson2016state} or \citealp{trentini2017measles}). 
	However, in recent years many people either refuse drastically any vaccination scheme or reduce (or delay) the prescribed vaccination. This has been often associated with pseudo-scientific beliefs.\footnote{\cite{dube2016vaccine}.}
	The phenomenon has become more pronounced in the last decades, especially in Western Europe and in the US,\footnote{%
	See  \cite{larson2016state} for a general and recent cross country comparison.
	Most studies are based on the US population: \cite{robison2012frequency}, \cite{smith2011parental}, \cite{nadeau2015vaccinating} and \cite{phadke2016association} are some of the more recent ones.
	\cite{funk2017critical} focuses on measles in various European countries.
	\cite{rey2018vaccine} analyzes the case of France.
	}
	and many public health organizations have issued public calls to researchers to enhance the understanding of the phenomenon and its remedies. 
		Even in the present times of Covid19 epidemic,  the opposition to vaccination policies is alive.\footnote{On this, see the recent reports of \cite{johnson2020online}, \cite{ball2020anti} and \cite{malik2020determinants}.
		}

		The focus of this paper is on the impact of \emph{homophily}, that is, the possible limitation of contacts between people with different pseudo-science attitudes, namely those in favor and those against vaccines. We study the effects of homophily levels on the dynamics of the disease, and the interaction of these levels with vaccination choices and with the popularity of anti--vaxxer movements.

There is a lot of evidence that, for example in the US, private and charter schools have a higher level of non-vaccinated children,\footnote{\cite{mashinini2020impact}, \cite{shaw2014united}.} and, in particular, a larger number of families that use the possibility of religious or philosophical exemptions.\footnote{\cite{zier2020attend}.}
	This phenomenon entails significant risks and, in order to protect the public, in many countries recent laws forbid enrollment of non-vaccinated kids into public schools. This is believed to have brought to an increase in enrollment in more tolerant private schools.\footnote{%
This phenomenon is documented for California by \cite{silverman2019lessons}. 
Recent evidence shows that similar trends \href{https://www.ilfoglio.it/salute/2019/05/23/news/chiuso-un-altro-asilo-no-vax-la-mappa-delle-scuole-fantasma-che-piacciono-agli-antivaccinisti-256645/}{happened in Italy} and \href{https://www.nytimes.com/2019/04/15/nyregion/measles-nyc-yeshiva-closing.html}{have been considered a cause of the measles outbreak in Manhattan in April 2019}.}%
$^{,}$%
\footnote{As another example, in light of the policies enforced during the Covid19 crisis, many companies and other public and private organizations have applied  \emph{rotation schemes} to limit physical interaction between people (on this, see the recent work by \citealp{ely2020rotation}): it is admissible that a policy maker may want to include non-vaccinated people all in the same group. }
These interventions, though, may as well affect the sorting of people with different pro or anti-vaccination attitudes in different schools. 
On an abstract level, this can corresponds to a change in the  \emph{homophily} of interactions, incentivizing people with anti-vaccination beliefs, who are, most likely, the ones with lower vaccination rates, to interact more together. This can have an important effect on the diffusion of epidemics but also on the formation of cultural norms.\footnote{For example, \cite{sobo2015social} argues that school community norms have an important impact in vaccine skepticism among families of children attending Steiner schools.} Moreover, during the Covid19 outbreak, governments have implemented very strong and drastic temporary containment and quarantine policies.
		However, such stringent policies cannot be permanent measures, and in normal times the policy makers are able to implement only milder policies that may segregate people in certain loci of activity.
	Limitations for attending schools are milder measures of this kind.

Moving from these premises, an important object of interest for policy-makers is the rate of contacts between two groups of people: those that are against vaccination and all the others, which we call for simplicity \emph{anti-vaxxers} and \emph{pro-vaxxers} (or just vaxxers), respectively. The two groups differ in their judgment about the real cost of vaccination, which is deemed higher by anti-vaxxers. This can be thought of as a psychological cost, a sheer mistake, or any phenomenon that may lead to a difference in perceived cost: we remain agnostic on the cause of it as our aim is to study its consequences. 
We think of this difference in the perceived cost as a basic cultural trait, which, as in the literature on cultural transmission, affects the  \emph{preferences} (or beliefs) of agents, who are still free to choose to vaccinate or not, based on a heterogeneous component of the cost. This means that an anti-vaxxer in our model may still vaccinate if the heterogeneous component of the cost is small enough.\footnote{For the sake of tractability, in the main text, we fully explore the corner-solution case in which all anti-vaxxers do not vaccinate. We explore the more general model with vaccine hesitancy in Appendix \ref{app:interior}, and we show that the main results carry through.} Through this, we mean to capture not the extremists that would never take a vaccine, but the much more general phenomenon of \emph{vaccine hesitancy}, which is much more widespread and, so, potentially much more dangerous \citep{trentini2017measles}. 

The homophily of contacts between the two groups is modeled by $h \in [0,1]$, which is the percentage of contacts that people cannot have with the other group (because, for example, their kids are not in the same schools, or they cannot meet in the same job and leisure places). We think $h$ as a number that is far from one (which would be the case of total segregation).
This biased pattern of contacts is in place before the epidemic actually takes place.
We show that more homophily may cause the disease to die out more slowly and cause more infection in the whole population, or even more infection among vaxxers.  In particular, our results suggest care both to a social planner concerned with total infection in the population \emph{and} to a social planner concerned only with infection among the vaxxers. The choice between the two approaches depends on the attitude toward society we want to model, and in particular on the specific interpretation of the difference in perceived cost, e.g., as a pure bias that the social planner should consider as such, or as a form of real psychological cost that we may want to factor in the welfare computation.\footnote{These are complex issues at the forefront of research in behavioral economics, see \cite{bernheim2018behavioral}.} As a consequence of these considerations, we remain agnostic on a general welfare criterion and explore instead the physical outcome of the amount of infection that, in such an environment, is likely to be a prominent, if not the only, element of any welfare analysis. 

The reason why an increase in $h$ may generate more infection is that homophily protects the group with fewer infected agents because it decreases the contacts and, thus, the diffusion of the disease across groups. Which group has a larger infection rate will, in turn, depend on initial conditions and on the difference in vaccination rates between the two groups. If the total number of agents initially infected is the same across the two groups,\footnote{This can happen, e.g., if the initial seeds are unequally distributed, and initially more vaxxers are infected, see Section \ref{cumulative}.} then homophily has no effect on total infection. Hence, a planner that cares only about the infection among vaxxers has no clear choice: she will desire an increase of $h$ (e.g., in case of an outbreak among anti-vaxxers) but would have opposite preferences in case of an outbreak among vaxxers.

If, instead, the two groups differ in the number of infected agents, the effect on total infection depends on the interplay of initial conditions and vaccinations. If the less vaccinated group happens to have more infections (because it suffered a larger share of the initial outbreak), we know homophily further increases infections in such group. The crucial observation is that it increases infections at a disproportionately larger rate than when the more vaccinated group has more infections. As a result, if the outbreak is among anti-vaxxers, total infection in the population increases with $h$, whereas if the outbreak is among vaxxers it decreases.

First, we consider a mechanical model in which vaccination choices are exogenous.
Then, we endogenize the vaccination choices of agents. Vaccination choices are taken before the disease spreads out. 
We view this as a classical trade-off between the perceived cost of vaccinating and the expected cost of getting sick.
In the model, the difference between anti--vaxxers and pro--vaxxers is only in the perceived costs of vaccination.
We show that even if we endogenize these choices, the qualitative predictions of the mechanical model are still valid: 
an increase in homophily is counterproductive.

Finally, we endogenize the choice of agents on whether to be anti--vaxxer or pro--vaxxer. This choice is modeled as the result of social pressure, with the transmission of a cultural trait.
There is a well-documented fact about vaccine hesitancy that seems hard to reconcile with strategic models: the geographical and social clustering of vaccine hesitancy. Various studies, reviewed e.g.~by \cite{dube2016vaccine}, find that people are more likely to have positive attitudes toward vaccination if their family or peers have. This is particularly evident in the case of specific religious confessions that hold anti-vaccination prescriptions and tend to be very correlated with social contacts and geographical clustering. These studies, though observational and making no attempt to assess causal mechanisms, present evidence at odds with the strategic model: if the main reason not to vaccinate is free riding, people should be \emph{less} likely to vaccinate if close to many vaccinated people, and not vice versa. In addition, \cite{lieu2015geographic} show that vaccine-hesitant people are more likely to communicate together than with other people. \cite{edge2019observational} document that vaccination patterns in a network of social contacts of physicians in Manchester hospitals are correlated with being close in the network. 
	It has also been shown that, in many cases, providing more information does not make vaccine-hesitant people change their minds (on this, see \citealp{nyhan2013hazards,nyhan2014effective} and \citealp{nyhan2015does}).
However, people do change their minds about vaccination schemes, as documented recently by  \cite{brewer2017increasing}, for example.
In a review of the literature, \cite{yaqub2014attitudes} finds that lack of knowledge is cited less than distrust in public authorities as a reason to be vaccine-hesitant. This is true both among the general public and professionals: in a study of French physicians, \cite{verger2015vaccine} finds that only 50\% of the interviewed trusted public health authorities. They both find a correlation between vaccine hesitancy and the use or practice of alternative medicine.

When we fully endogenize the choices of agents (both membership to groups and vaccination choices), we find that the predictions of the simple mechanical model remain valid only if the groups of the society are rigid enough, and people do not change their minds easily about vaccines.
If, instead, people are more prone to move between the anti--vaxxers and pro--vaxxers groups, then segregation policies can have positive effects.
The simple intuition for this is that, when anti--vaxxers are forced to interact more together, they internalize the higher risk of getting infected and, as a result, they are more prone to become pro--vaxxers.



\bigskip

We contribute to three lines of literature, related to three steps of our analysis highlighted above: the analysis of the effects of segregation in epidemiological models, the economics literature on vaccination and its equilibrium effects, and the literature on diffusion of social norms and transmission of cultural traits. 

The medical and biological literature using SI-type models is wide, and a review of it is beyond our scope. We limit ourselves to note that recently some papers have considered dynamic processes with formal similarity to ours. 
\cite{jackson2013diffusion} and \cite{izquierdo2018mixing} are the first, to our knowledge, to study how homophily affects diffusion.
\cite{pananos2017critical} analyze critical transitions in the dynamics of a three equation model including epidemic and infection.

The literature on strategic immunization has analyzed models where groups are given and the focus is the immunization choice, as in \cite{galeotti2013strategic}, or both the immunization and the level of interaction are endogenous, as in \cite{goyal2015interaction}. \cite{chen2014economics} argue that the market mechanism yields inefficiently low levels of vaccination, while \cite{talamas2020free} show that a partially effective vaccination can decrease welfare, with a mechanism that, like ours, works via behavioral responses.
At an abstract level, the difference with respect to our setting is that we endogenize the group partition through the diffusion of social norms.\footnote{%
There is also a recent literature in applied physics that studies models where the diffusion is simultaneous for the disease and for the vaccination choices. On this, see the review of \cite{wang2015coupled}, and the more recent analysis of \cite{alvarez2017epidemic} and \cite{velasquez2017interacting}. 
} 

The economics of social norms and transmission of cultural traits is a lively field, surveyed by \cite{bisin2011economics}.  Common to this literature is the use of simple, often non-strategic, dynamic models of the evolution of preferences. We adopt this framework, finding it useful despite the differences we discuss later. A paper close to ours is \cite{panebianco2017paternalism}, which considers how social networks affect cultural transmission in a SI-type model, with a more concrete network specification through degree distributions. The literature on segregation in cities and communities has also studied the trade-offs generated by stratification and asymmetric interactions, and the inefficiencies of social separation: cfr \cite{benabou1993workings,benabou1996equity,benabou1996heterogeneity}.

	\bigskip
	
		The paper is organized as follows.
		Next section presents the model.
		Section \ref{sec:mechanical} shows results for the mechanical model, when all choices are exogenous.
		Sections \ref{sec:vaccination} and \ref{sec:groups} introduce respectively endogenous vaccination and endogenous group membership, deriving our analytical results for these cases.
		We conclude in Section \ref{sec:conclusion}.
		In the appendices we consider a microfoundation of the cultural transmission mechanism (Appendix \ref{app:cultural}), extensions of the model (Appendices \ref{app:interior} and \ref{sec:extensions}) 
		and we prove our formal results (Appendix \ref{app:proofs}).
		

\section{The Model}
	
\label{sec:model}

\subsection{The Epidemic}	
	
We consider a simple SIS model with vaccination and with two groups of agents, analogous to the setup in \cite{galeotti2013strategic}. To understand the main forces at play, we start by taking all the decisions of the agents as exogenous, and we focus on the infection dynamics.
In the following sections, we endogenize the choices of the players.

Our society is composed of a continuum of agents of mass $1$, partitioned into two groups.
To begin with, in this section this partition is exogenous.
Agents in each group are characterized by their attitude towards vaccination. In details, following a popular terminology, we label the two groups with $a$, for \emph{anti-vaxxers}, and with $v$, for \emph{vaxxers}. Thus, the set of the two groups is $G:=\{a,v\}$, with $g\in G$ being the generic group. 
Let $q^a\in[0,1]$ denote the fraction of \emph{anti-vaxxers} in the society, and $q^v=1-q^a$ the fraction of \emph{vaxxers}.
To ease the notation, we write $q$ for $q^a$, when this does not create ambiguity.

People in the two groups meet each other with an \emph{homophilous} bias. We model this by assuming that an agent of any of the two groups has a probability $h$ to meet someone from her own group and a probability $1-h$ to meet someone else randomly drawn from the whole society.\footnote{$h$ is the \emph{inbreeding homophily} index, as defined in \cite{coleman1958relational}, \cite{marsden1987core}, \cite{mcpherson2001birds} and \cite{currarini2009economic}. It can be interpreted in several ways, as an outcome of choices or opportunities.
As we assume that $h$ can be affected by groups' choices and by policies, we can interpret it as the amount of time in which agents are kept segregated by group, while in the remaining time they meet uniformly at random.  
}
This implies that anti-vaxxers meet each others at a rate of $\tilde{q}^a:=h+(1-h)q^a$, while vaxxers meet  each others at a rate of  $\tilde{q}^v:=h+(1-h)q^v=h+(1-h)(1-q^a)$. Note that $h$ is the same for both groups, but if $q^a \ne q^v$ and $h>0$, then $\tilde{q}^a\neq \tilde{q}^v$.

For each $g\in G$, let $x^g\in[0,1]$ denote the fraction of agents in group $g$ that are vaccinated against our generic disease. It is natural to assume, without loss of generality, that $x^a<x^v$, and by now this is actually the only difference characterizing the two groups. Let $\mu$ be the recovery rate of the disease, whereas its infectiveness is normalized to $1$. 

\subsubsection{The dynamical system}

Setting the evolution of the epidemic in continuous time, we study the fraction of infected people in each group. When this does not generate ambiguity, we drop time indexes from the variables. For each $i\in G$, let $\rho^i$ be the share of infected agents in group $i$. Since vaccinated agents cannot get infected, we have $\rho^a\in[0,1-x^a]$ and $\rho^v\in[0,1-x^v]$, respectively. 

The differential equations of the system are given by:
\begin{eqnarray}
\label{system1}
\dot{\rho}^a & = & \big(1-\rho^a-x^a \big) \Big( \tilde{q}^a \rho^a + (1-\tilde{q}^a) \rho^v \Big) - \rho^a \mu;
\nonumber \\
\dot{\rho}^v & = & \big(1-\rho^v-x^v \big) \Big(  \tilde{q}^v\rho^v + (1-\tilde{q}^v) \rho^a \Big) - \rho^v \mu  .
\end{eqnarray}
For each $g\in G$, $\big(1-\rho^g-x^g \big)\in[0,1]$ represents the set of agents who are neither vaccinated, nor infected, and thus susceptible of being infected by other infected agents. Moreover, the share of infected agents met by vaxxers and anti-vaxxers is given by $\Big( \tilde{q}^a \rho^a + (1-\tilde{q}^a) \rho^v \Big)$ and by $\Big(  \tilde{q}^v\rho^v + (1-\tilde{q}^v) \rho^a \Big)$, respectively. Finally, $\rho_a \mu$ and $\rho_v \mu$ are the recovered agents in each group. 

We are going to assume that at the beginning of the epidemic a fraction of agents is infected, selected at random independently of the group. We can think for example of random encounters with spreaders coming from another country or region. Since the initial infected status is independent of group identity, the initial condition is symmetric: $\rho^a_0=\rho^v_0=\rho_0$.\footnote{An alternative is to think of the different fractions $\rho^a_0$ and $\rho^v_0$ each extracted at random from distributions with the same mean $\EE \rho^a_0=\EE \rho^v_0=\rho_0$. This will not change our results because in the following we will focus on the linearization around the steady state, so our expressions will depend linearly on the initial conditions.}

\begin{prop}[{Homophily and endemic disease}]
\label{SS}
The system \eqref{system1} always admits a trivial steady state: $ ( \rho^{a}_1 , \rho^{v}_{1}):=(0,0)$. 
For each $h$, there exists a $\hat{\mu}(h)>0$ such that (i) if  $\mu<\hat{\mu}(h)$, $(0,0)$ is unstable, whereas (ii)  if $\mu>\hat{\mu}(h)$,  $(0,0)$ is stable.\footnote{%
\label{note_Delta}
Note that $\hat{\mu}(h):=\frac{1}{2}\left(T+\Delta\right)\in[0,1]$, where $T:=\tilde{q}^a(1-x^a)+\tilde{q}^v(1-x^v)$ and $\Delta:=\sqrt{T^2-4h(1-x^a)(1-x^v)}$.	$\Delta$ is always positive and it is  increasing in $q$.
Moreover $\hat{\mu}(h)\in[0,1]$ and its value is 
$ 1 - x^v + q (x^v - x^a )  $
for $h=0$ and 
$
1-x^a
$
for $h \rightarrow 1$. }
\end{prop}

This result is obtained in the standard way, by setting to zero the two right--hand side parts of the system in \eqref{system1} and solving for $\rho^a$ and $\rho^v$.
The formal passages are in Appendix \ref{app:proofs}, as those of the other results that follow.

In the remaining of the paper, we focus on the case in which $\mu>\hat{\mu}(h)$, because it is consistent with diseases that are not endemic but show themselves in episodic or seasonal waves.
For those diseases, society lays for most of its time in a steady state where no one is infected. However, exogenous shocks increase the number of infected people temporarily. Eventually, the disease dies out, as it happens, for example, for the seasonal outbreaks of flu.

Note that $\hat{\mu}(h)$ is increasing in $h$, so that we can highlight a first important role for $h$ in the comparative statics. If $h$ increases, it is possible that a disease that was not endemic, because $\mu > \hat{\mu}(h)$, becomes so because $\hat{\mu}(h)$ increases with $h$, and the sign of the inequality is reversed. Indeed, higher homophily counterbalances the negative effect that the recovery rate $\mu$ has on the epidemic outbreak.

\subsubsection{Cumulative Infection}
The main focus of our interest is to see what is the welfare loss due to the epidemic, and how this depends on the policy parameter $h$.  In our simple setting, the welfare loss is measured by the total number of infected people over time, that is \textbf{cumulative infection}.
For analytical tractability, we will approximate the dynamics of outbreaks with the linearized version of the dynamics $\hat{\rho}$, that satisfies:
\begin{align}
    \dot{\hat{\rho}}_t&=\mathbf{J}\left(\begin{array}{c}
        \hat{\rho}^a_t  \\
        \hat{\rho}^v_t 
    \end{array}\right)\ ,\  \quad \hat{\rho}_0=\left(\begin{array}{c}
        \rho^a_0  \\
        \rho^v_0 
    \end{array}\right),
\end{align}
where $\mathbf{J}$ is the Jacobian matrix of \eqref{system1} calculated in the $(0,0)$ steady state, and $(\rho^a_0,\rho^v_0)'$ is the initial magnitude of the outbreak. 

The cumulative infection in the two groups and in the overall population is (approximately, for a small perturbation around the steady state):
\begin{align}
CI^a&:=\int_0^{\infty}\hat{\rho}^a(t)d t,  \nonumber\\
CI^v&:=\int_0^{\infty}\hat{\rho}^v(t)d t,  \label{cumulative} \\
CI&:=q^aCI^a+(1-q^a)CI^v. \nonumber
\end{align}

Note that, since $q^a$ is fixed, $CI$ takes into account both the number of infected agents of each group at each period and also the length of the outbreak. In the range of parameters for which $(0,0)$ is stable, all the integrals are finite, so here we do not add discounting, for simplicity. We will explore the implications of introducing time preferences in Section \ref{discounting}. The expressions are:
\begin{align}
CI^a& =\rho_0\frac{2 \left( \left(\mu -(1-x^v)\tilde{q}^v\right)+ \left(1-x^a\right)
	\left(1-\tilde{q}^a\right)\right)}{(-\Delta -2 \mu +T) (\Delta -2 \mu +T)}, \label{CIa} \\
CI^v& =\rho_0\frac{2 \left( \left(1-x^v\right) \left(1-\tilde{q}^v\right)+ \left(\mu -(1-x^a)\tilde{q}^a\right)\right)}{(-\Delta -2 \mu +T) (\Delta -2 \mu +T)}, \label{CIv} \\
CI&=2\rho_0\frac{\mu-(1-x^a)(\tilde{q}^a-q)-(1-x^v)(\tilde{q}^v-1+q) }{(-\Delta -2 \mu +T) (\Delta -2 \mu +T)}. \label{CI} 
\end{align}

How good is the above approximation using the linearized dynamics? Theory implies that the linear approximation is good in a neighborhood of the steady state, for small values of $\rho^a_0$ and $\rho^v_0$. In Figure \ref{approximation} we depict, for comparison, the trajectories of $\rho^a$ and $\rho^v$ numerically calculated from the original nonlinear system and the linearized approximation. We use on purpose an extremely large value of the initial conditions: $\rho^a_0=\rho^v_0=0.5$ (namely 50\% of the population is infected at the beginning). We can see from the graphs that the curves are very similar and close to each other even in this extreme case, and for a large range of values of homophily $h$. This suggests that for the simple SIS model that we study, the qualitative behavior of the linear approximation is very close to the actual solution.

\begin{figure}
  \centering
\includegraphics[width=0.45\textwidth]{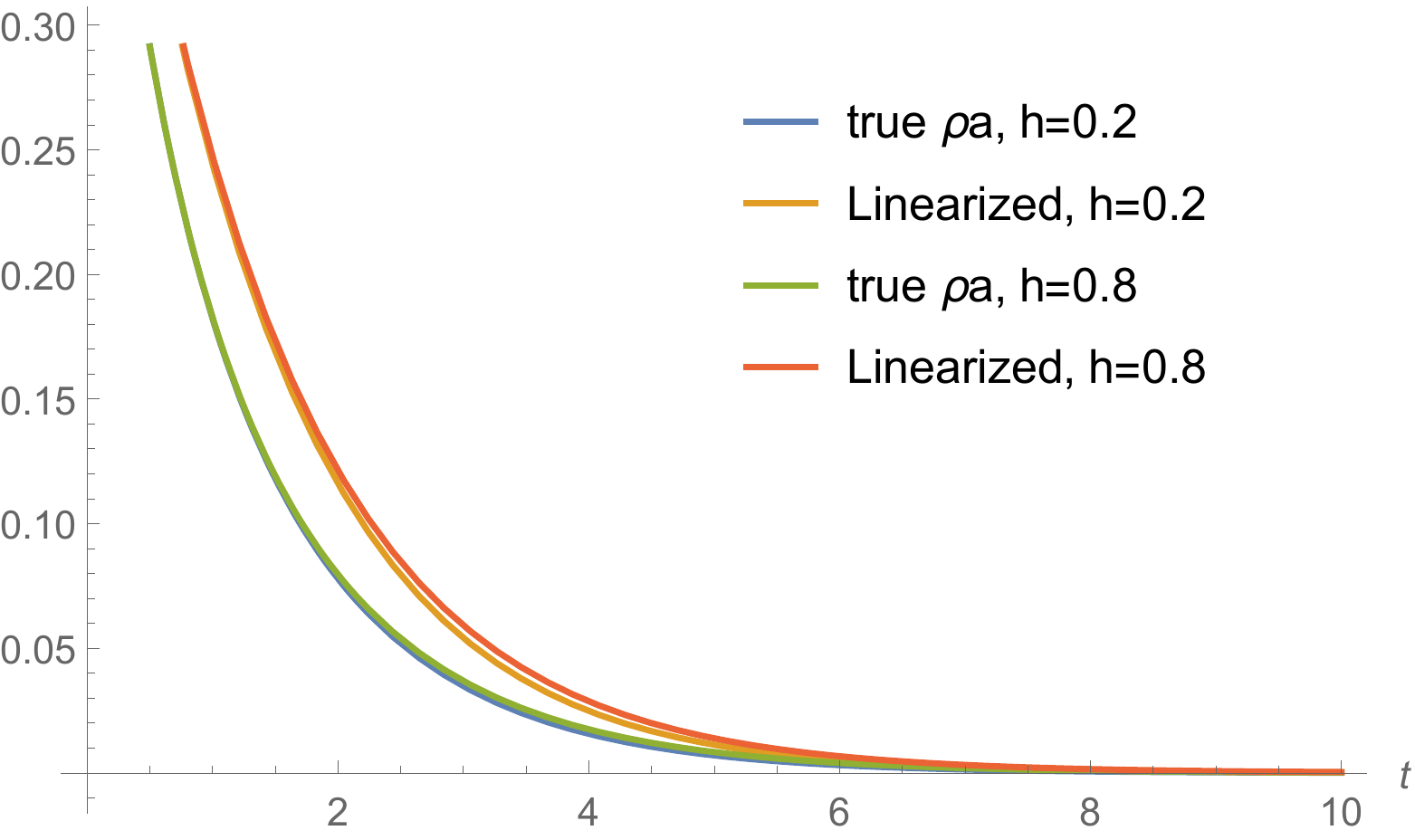}
\includegraphics[width=0.45\textwidth]{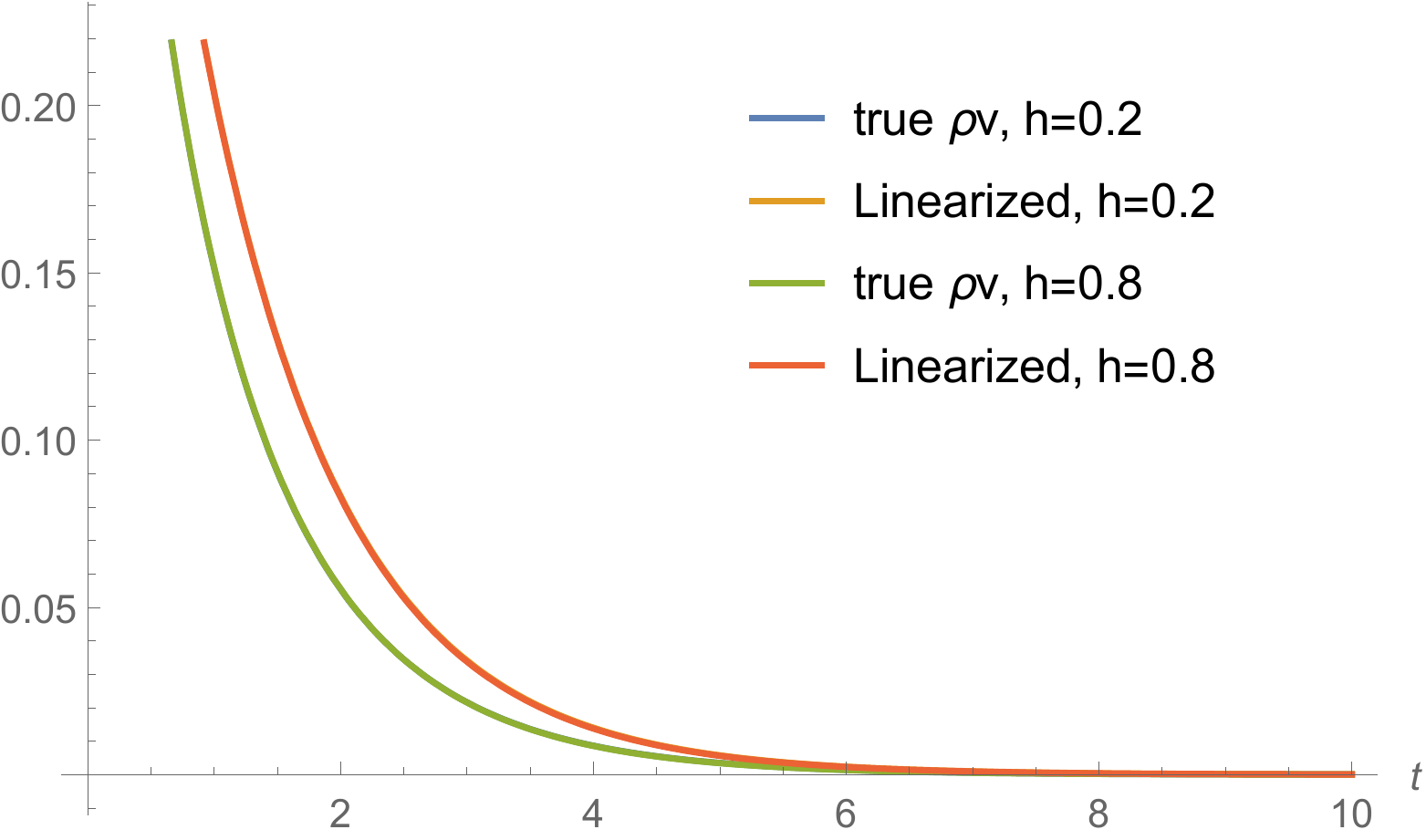}
\caption{{\bf Left panel}: $\rho^a$ as a function of time, actual solution and linearized, for $h=0.2$ and $h=0.8$. 
{\bf Right panel}: $\rho^v$ as a function of time, actual solution and linearized, for $h=0.2$ and $h=0.8$. The other parameters are set at $\mu=1$, $x^a=0.7$, $x^v=0.9$, $\rho^a_0=\rho^v_0=0.5$.
}
\label{approximation}
\end{figure}

\subsection{Vaccination choices}
	
	In the second step of our analysis, we endogenize vaccination choices.
We assume that agents take vaccination decisions ex-ante, before an epidemic actually takes place, and cannot update their decision during the diffusion.
This mimics well diseases, like seasonal flu, for which the vaccine takes a few days before it is effective, and the disease spreads rapidly among the population.
Agents take the decision considering the trade-off between paying some fixed cost for vaccinating or incurring the risk of getting infected, thus paying with some probability a cost associated with health.\footnote{%
See, for example, \cite{bricker2019postmodern} and \cite{greenberg2019measles} for a recent analysis of the anti--vaxxers arguments: 
Those are mostly based on conspiracy theories that attribute hidden costs to the vaccination practice and not so much on minimizing the effects of getting infected. 
Our model would not change dramatically if we attribute the difference in perception on the costs of becoming sick,
but we stick to the first interpretation because it makes the computations cleaner.}\\

{\bf Vaccination costs} For the reasons discussed in the introduction, we do not aim to microfound the discrepancy in the evaluation of vaccinations costs between vaxxers and anti-vaxxers. Hence, with a descriptive spirit, we adopt the assumption that anti-vaxxers have a cost larger than vaxxers of a uniform amount $d$. To be precise, we assume that, for vaxxers, vaccination costs are $c^v\sim U[0,1]$, whereas for anti-vaxxers $c^a\sim U[d,1+d]$.

\bigskip

{\bf Infection risks} Let us denote by $\sigma^i(x^a,x^v,q,h)$ the function of the parameters indicating the expected welfare cost of infection for an agent in group $i$, or equivalently her estimation of the risk from non being vaccinated. Given the distributional assumption on the cost made above, it follows that the fraction of people vaccinating in each group is equal to the perceived welfare loss from the risk of infection. That is, in equilibrium, $x^a,x^v$ satisfy:
\begin{align}
x^a&=\max\{\sigma^a(x^a,x^v,q,h)-d,0\},\\
x^v&=\sigma^v(x^a,x^v,q,h).
\end{align}

The functional form of $\sigma$ can be specified in different ways, according to how agents measure the risk of infection. Some of our results depend only on general assumptions on the behavior of $\sigma$, while others need an analytical specification. First we discuss the general assumption we maintain and, subsequently, we present two examples of functional forms that we will use throughout the paper.

We adopt the following high-level assumption:
\begin{assumption}
   \label{vax1} 
   Agents perceive a larger risk of infection if they have less vaccinated neighbors.
    
    
\end{assumption}

What does Assumption \ref{vax1} imply for the functional form of $\sigma^i$? Our mean field dynamics for social contacts implies that the fraction of vaccinated neighbors an agent in group $i$ meets is $\tilde{q}^ix^i+(1-\tilde{q}^i)x^j$. Following Assumption \ref{vax1}, $\sigma^i$ should increase whenever this quantity increases. This implies the more concrete conditions:
\begin{enumerate}[i)]

\item $\sigma^i$ is bounded, non-negative, and differentiable;
    
    \item $\sigma^i$ is decreasing in $x^i$ and $x^j$ and $\sigma^i(1,1,q,h)=0$ (\textit{positive externality of vaccination});
    
    \item if $x^i>x^j$, then $\sigma^i$ is increasing in $h$, otherwise is decreasing in $h$ (\textit{homophily favors the more vaccinated group});
    
    \item if $x^a>x^v$, then $\sigma^i$ is increasing in $q$ (\textit{negative externality of anti--vaxxers}). 

\end{enumerate}

Depending on the specific application, group $a$ in this model can capture two types of people: complete \emph{vaccine skeptical}, who never vaccinate, or \emph{vaccine hesitant}, who hold a higher estimation of costs, but might be willing to vaccinate anyway. The first situation can be captured in the case in which $d$ is large enough so that in equilibrium no anti--vaxxer wants to vaccinate, that is $x^a=0$. Let us call $\underline{d}$ a threshold for $d$ such that if $d>\underline{d}$ no anti--vaxxers want the vaccine. Such $\underline{d}$ always exists, provided $\sigma^a$ is bounded. In the main text we focus on such equilibrium with \emph{extreme} anti--vaxxers, that allows the sharper analytical characterizations. We defer to Appendix \ref{app:interior}
 the discussion of the case of a milder bias such that $d<\underline{d}$, that is the case of interior equilibria for anti--vaxxers (\emph{vaccine hesitancy}). So, throughout the main text of the paper, we are going to maintain the following assumption:
\begin{assumption}[Extreme anti--vaxxers]
\label{corner}

$d> \overline{d}$.

\end{assumption}

\subsection{Examples}
\label{examples}

Two assumptions on $\sigma$ that satisfy the above assumptions, and balance simplicity and intuitive appeal are: $\sigma$ is proportional to the number of non-vaccinated; and $\sigma$ proportional to the cumulative infection. Our results until we endogenize group structure are general and do not depend on the functional form chosen for $\sigma$. However, in the endogenous group case (Section \ref{sec:groups}) we focus on the two possibilities discussed here.

\bigskip

{\bf Risk of infection proportional to non-vaccinated} (In the following, $NV-$risk). In this case we assume that agents think about the risk of infection using a simple heuristic: they estimate it as being proportional to the fraction of non-vaccinated people that they meet.
Agents multiply this fraction of non-vaccinated people by a factor $k>0$, that represents the perceived damage from the disease, which is the same for the two groups. Thus:
\begin{equation}
\label{sigmav}
\sigma^v = k[\tilde{q}^v (1-x^v) + (1-\tilde{q}^v) (1-x^a)] \   ,
\end{equation} 
 and similarly:

\begin{equation}
\label{sigmaa}
\sigma^a = k[\tilde{q}^a (1-x^a) + (1-\tilde{q}^a) (1-x^v)]\   ,
\end{equation}

The big advantage of this form is that we can easily solve for the fraction of vaccinated, obtaining:
\begin{align}
\label{endovax}
x^a&=0\ , \nonumber\\
x^v&=\frac{k}{(h-1) k q+k+1}\ ,
\end{align}
provided $d<1/k$ so that $x^v\neq0$.\footnote{This is possible if $\frac{h k^2+k}{h k q-k q+k+1}<d<\frac{1}{k}$ and either $k<1$
 or 	$ \left(1<k<\frac{1}{2}
	\left(1+\sqrt{5}\right)\land 0<q<\frac{-k^2+k+1}{k}\land 0<h<\frac{-k^2-k q+k+1}{k^3-k q}\right)$} 

\bigskip

{\bf Risk of infection measured by cumulative infection} (in the following we abbreviate with $CI-$risk) Another possibility is that agents evaluate the risk of infection using the cumulative infection in their respective group. In this case: 
\begin{align}
\sigma^a&=CI^a,\\
\sigma^v&=CI^v.
\end{align} 

Since this is the measure of the aggregate cost of infection, this example can capture a situation in which agents' assessment of the risk derives from the signals dispensed by a central authority. We can think about agents that do not independently collect and evaluate information, but instead delegate to the suggestions coming from the central authority the evaluation of the risk level. Since the central authority cares about the cumulative infection, so do the agents in turn.

\subsection{Endogenous groups}

In the third step of our analysis, we model how the shares of anti-vaxxers, $q$, is determined. In the real world, this decision does not seem to be updated frequently, and can be considered as fixed during a single flu season. So, in the model, we assume that this decision is taken before actual vaccination choices, which are in turn taken before the epidemic eventually starts.
Our aim here is to offer a simple and flexible theory of the diffusion of opinions to be integrated into our main epidemic model. The empirical observations that important drivers of vaccination opinions are peer effects and cultural pressure leads us to discard purely rational models, where the decision of not vaccinating descends only from  strategic considerations. Given the complex pattern of psychological effects at play, we opt for a simple reduced-form model capturing the main trade-offs. In particular, we are going to assume the diffusion of traits in the population to be driven by \emph{expected advantages}: the payoff advantage that individuals in each group estimate to have with respect to individuals in the other group. This is made precise in what follows. 

\paragraph{Socialization payoffs as expected advantage}
Consider an individual in group $a$. Define the \emph{socialization payoff} for group $a$, $\Delta U^a$, as the \emph{Expected advantage} individual $a$ estimates to have with respect to individuals in group $v$. Specifically:
\begin{align}
\Delta U^a&=U^{a a}-U^{a v}, \label{DUa} \\
U^{a a}&= -\EE_c^a\left[(c+d)\mathbbm{1}_{\sigma^a-d>c}+\sigma^a\mathbbm{1}_{\sigma^a-d\le c}\right], \\
U^{a v}&= -\EE_c^a\left[(c+d)\mathbbm{1}_{\sigma^v>c}+\sigma^v\mathbbm{1}_{\sigma^v\le c}\right], 
\end{align}
where $U^{a a}$ is the payoff of individuals with trait $a$ evaluated by an individual with trait $a$, while $U^{a v}$ is the payoff of individuals with trait $v$ evaluated by individuals with trait $a$. 

The socialization payoff $\Delta U^v$ is defined analogously:
\begin{align}
\Delta U^v&=U^{v v}-U^{v a}, \label{DUv} \\
U^{v v}&= -\EE_c^v\left[c\mathbbm{1}_{\sigma^v>c}+\sigma^v\mathbbm{1}_{\sigma^v\le c}\right], \\
U^{v a}&= -\EE_c^v\left[c \mathbbm{1}_{\sigma^a-d>c}+\sigma^a\mathbbm{1}_{\sigma^a-d\le c}\right]. 
\end{align}

Agents in each group perceive a differential in expected utilities from being of their own group as opposed to being of the other group. Note that, apart from the bias $d$, agents correctly evaluate all other quantities, including the risks from the disease of the two groups, $ \sigma^v$ and $ \sigma^a$. Indeed, even if both groups evaluate the choice of the other group as suboptimal, 
this perceived difference can be negative for \emph{anti--vaxxers}, because they understand that \emph{vaxxers} have less chances of getting infected.

Under Assumption \ref{corner} ($d> \underline{d}$ so that $(x^a)^*=0$), integration yields:
\begin{align}
\label{socpayoff}
\Delta U^a&=\sigma^v-\sigma^a+d \sigma^v-\frac{1}{2}(\sigma^v)^2,\\
\Delta U^v&=\sigma^a-\sigma^v+\frac{1}{2}(\sigma^v)^2.
\end{align}

To understand how the socialization payoffs are affected by infection, first notice that the risk of infection in the own group decreases the payoff, whereas the risk of infection in the other increases it. This captures the fact that, \textit{ceteris paribus}, high infection is undesirable. 

To clarify the definition of socialization payoffs, consider Figure \ref{fig:payoffs}. The black line is the disutility of agents in groups $a$, \emph{as perceived by agents in group $a$}, as a function of the cost $c$. As a consequence of Assumption \ref{corner}, it is constant and it does not depend on $c$ because, in the case we are focusing on, no agent in group $a$ vaccinates\footnote{In  Appendix \ref{app:endgroup} we depict the same graph in the case in which Assumption \ref{corner} does not hold, namely $(x^a)^*>0$}. 
The grey area below this curve is then $U^{a a}$. 
Consider now the red line that represents the disutility of agents in group $v$ \emph{as perceived by agents in group $a$}. In particular, agents in group $v$ have a different perception of costs with respect to agents in group $a$, and so take different choices. In particular, they vaccinate in the  $[0,\sigma^v]$  interval, whereas in the $[\sigma^v,1]$ interval they do not vaccinate and incur a risk of infection. Note, however, that this is the evaluation from the perspective of agents in group $a$, and thus the cost of vaccination is $c+d$ instead of $c$. Hence $ U^{a v}$ is the area below the red curve. The difference $\Delta U^{a}$ is given by the red area minus the blue area. $ U^{v v}$ and $ U^{v v}$ are computed accordingly. 

\begin{figure}[h]
    \centering

\begin{tikzpicture}[
thick,
>=stealth',
dot/.style = {
draw,
fill = white,
circle,
inner sep = 0pt,
minimum size = 4pt
}
]
\coordinate (O) at (0,0);
\draw[->] (-0.3,0) -- (8,0) coordinate[label = {below right:cost $c$}] (xmax);
\draw[->] (0,-0.3) -- (0,5) coordinate[label = {right: disutility}] (ymax);

\fill[gray!20] (0,0) -- (0,2.5) -- (7,2.5)  -- (7,0)  ;
\draw[dashed] (0,0) -- (2,2);
\draw[gray, <->] (-1,0) -- node[left] {$d$} (-1,3);

\draw [dashed] (2,0) -- (2,5);

\draw[dashed] (2.5,2.5)--(0,2.5) node[left] {$\sigma^a$};
\node[below] at (7,0) {$1$} ;   
   
\node[below] at (2,0) {$\sigma^v$} ;   

\draw[dashed] (2,2)--(0,2) node[left] {$\sigma^v$};	
\fill[red!20] (0,2.5)--(0,3)-- (2,5) -- (2,2.5);
\fill[blue!20] (2,2) -- (7,2) --  (7,2.5) --(2,2.5) ;   
\draw[red] (0,3) -- (2,5) -- (2,2)-- (7,2) node[right]{Disutility of $v$ as perceived by $a$};
\draw (0,2.5) --  (7,2.5) node[right]{Disutility of $a$ as perceived by $a$}; 

\end{tikzpicture}

    \caption{Composition of $\Delta U^a$. The graph represents the disutility incurred by an individual as a function of its cost $c$. $\Delta U^a$ is the red area minus the blue area. }
    \label{fig:payoffs}
\end{figure}
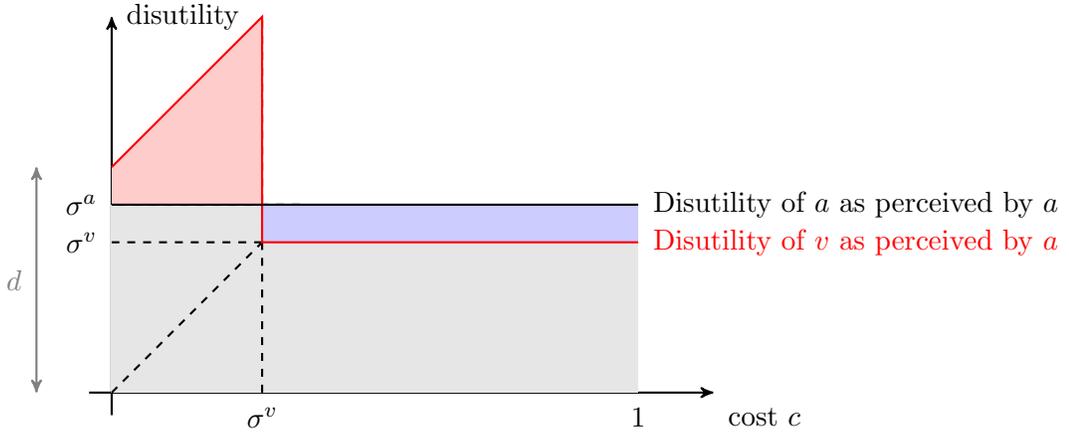

We now consider the population dynamics and, we make the following assumption:

\begin{assumption}
\label{simple_q}
Given an $\alpha\in \mathbb{R}$, the level of $q$ increases when $q^{\alpha}\Delta U^a > (1-q)^{\alpha}\Delta U^{v}$ and it decreases when $q^{\alpha}\Delta U^{a} < (1-q)^{\alpha}\Delta U^{v}$.
\end{assumption}

Clearly, the implication of the previous assumption is that the resting points of the dynamics are such that $q^{\alpha}\Delta U^{a} = (1-q)^{\alpha}\Delta U^{v}$, but stability has to be addressed. The simplest example of dynamics satisfying Assumption \ref{simple_q} is:
\[
    \dot{q}=q(1-q)[q^\alpha \Delta U^a-(1-q)^\alpha\Delta U^v] \ \ ,
\]
but we allow also for any non linear generalization.

The dynamics obtained from Assumption \ref{simple_q} generalizes the standard workhorse model in cultural transmission, the one by \cite{bisin2001economics}, in two ways: $(i)$ endogenizing the socialization payoffs and $(ii)$ introducing a parameter $\alpha$ regulating the \emph{stickiness} agents have in changing their identity \emph{via} social learning. Indeed, at the limit $\alpha \rightarrow \infty$, $\dot q=0$ and types are fixed. Note also that $\alpha$ regulates the strength of \emph{cultural substitution}, a phenomenon often observed in cultural transmission settings: the tendency of members of minorities to preserve their culture by exerting larger effort to spread their trait.\footnote{See \cite{bisin2001economics, bisin2011economics}.} Thus, we are able to  encompass different types of social dynamics. $(i)$ If $\alpha=0$, this is a standard replicator dynamics (see e.g.~\citealp{weibull1997evolutionary}). $(ii)$ If $\alpha<0$, the model displays cultural substitution, as most standard cultural transmission models. Moreover, the more $\alpha$ is negative, the more there is substitution. In particular, if $\alpha=-1$ the dynamics has the same steady state and stability properties as the dynamics of \cite{bisin2001economics}.\footnote{To be precise, the model by \cite{bisin2001economics} refers to \emph{intergenerational} transmission. In Appendix \ref{app:cultural} we show how a similar equation can be recovered in a context of \emph{intragenerational} cultural transmission} $(iii)$ If $\alpha>0$, the model displays cultural complementarity, so that the smaller the minority the less the minority survives. Note that cultural complementarity is increasing in $\alpha$.

Note that the environment of social influence is not only shaped by physical contacts and it is not the same of the epidemic diffusion of the actual disease (because in the real world many contacts are online and are channeled by social media). Hence, any policy on $h$ can have a limited effect on it, because for us $h$ is a restriction on the physical meeting opportunities. As a consequence, $h$ does not appear explicitly in Assumption \ref{simple_q}.

\section{The Epidemic}
\label{sec:mechanical}

In this section we start analyzing the pure epidemic part of the model, taking the vaccination rates $x^a$ and $x^v$, and the share of anti--vaxxers $q$ as exogenous.
Remember that in this case the only difference between the two groups is that $x^a<x^v$.

Which group has more infected agents throughout the epidemic? An immediate calculation using expressions \eqref{CIa} and \eqref{CIv} yields that $CI^a \ge CI^v$ or, more in general, the largest infection is in the group with the smallest fraction of vaccinated agents.



In particular, the evaluation of which group is better off in terms of infections is \emph{independent} of homophily. However, the levels of contagion do depend on homophily, as the following result shows, which is obtained applying definitions from the expressions in \eqref{cumulative} and taking derivatives.

\begin{prop}[Effect of $h$ and $q^a$]
\label{prop:exogenoush}
\
Under exogenous vaccination choices:
\begin{enumerate}[a)]
\item $CI$ and $CI^a$ are increasing in $h$; $CI^v$ is decreasing  in $h$;
\item $CI$, $CI^a$ and $CI^v$ are decreasing in $x^v$ and $x^a$;
\item $CI$, $CI^a$ and $CI^v$ are increasing in $q$.
\end{enumerate}

\end{prop}

The effect of vaccination rates on CI is the expected one: more vaccinated agents mean lower infection levels. Similarly, an increase in the number of anti-vaxxers $q$ means an increase in the number of non-vaccinated agents, so for an analogous reason it increases infections in all groups.

 Note first that homophily $h$ has a \emph{redistributive} effect: it protects the group with more vaccinated and, in our case, group $v$. As members of group $v$ are less likely to meet members of group $a$, their risk of infection decreases, so (since so far we maintain $x^v$ exogenously fixed) their infection level decreases. The symmetric happens for members of group $v$.

However, homophily also has a cumulative effect, increasing the number of total infections, $CI$. The intuition behind this is that homophily increases the time that the epidemic takes to go back to the zero steady state.\footnote{This is common in dynamic problems: \cite{golub2012homophily} find a similar effect in a learning setting.}

To clarify this point, we consider as a measure of convergence time the magnitude of the leading eigenvalue, which in this case is the one with the smallest absolute value. This is because
the solution of our linear system is a linear combination of exponential terms whose coefficients are the eigenvalues (which are negative by stability). Hence, when $t$ is large, the dominant term is the one containing the eigenvalue which has smallest absolute value.\footnote{%
 We should be careful, though, because this is true non--generically outside of the \emph{eigendirection} of the \emph{second} eigenvector. Indeed, in our case the eigenvectors are:
\[
\boldsymbol{e}_1=\left(-\frac{\left(1-x^v\right) \tilde{q}^a+\left(x^a-1\right) \tilde{q}^a+\Delta }{2 \left(1-x^v\right) \left(1-\tilde{q}^a\right)},1\right) \ \ ,
\]
and 
\[
\boldsymbol{e}_2=\left(\frac{-\left(1-x^v\right) \tilde{q}^a-\left(x^a-1\right) \tilde{q}^a+\Delta }{2 \left(1-x^v\right) \left(1-\tilde{q}^a\right)},1\right) \ \ .
\]
So, we can see that the first eigendirection does not intersect the first quadrant, while the second does. Hence, we should remember that the first eigenvalue is a measure of the speed of convergence only generically, outside of the eigendirection identified above.}

\begin{prop}
\label{convtime}
Under exogenous vaccination choices, consider a perturbation around the stable steady state $(0,0)$. The time of convergence (as measured by the leading eigenvalue) back to $(0,0)$ is increasing in $h$.
\end{prop}

This result shows that homophily, by making the society more segregated, makes the convergence to the zero infection benchmark slower once an outbreak occurs. This is obtained by analyzing the eigenvalues of the Jacobian matrix, computed in the steady state. 
All results are obtained analytically (see Appendix \ref{app:proofs}), and the resulting eigenvalues are decreasing in absolute value in $h$.

If we look at the effects of other parameters, we have that the eigenvalues are increasing in absolute value in both $x^a$ and $x^v$. This is because a larger number of vaccinated agents means a smaller space for infection to diffuse. 
Finally, since $x^a<x^v$, then the smallest eigenvalue is decreasing (in absolute value) in $q^a$, while the largest eigenvalue is increasing. Since the long-run dynamics (i.e.~asymptotic convergence) depends on the smallest eigenvalue, this means that the dynamics is \emph{asymptotically slower} the larger the fraction of the population with less vaccinated agents.



To sum up, Propositions \ref{system1}, \ref{prop:exogenoush} and \ref{convtime} provide clear implications that should be taken into account when considering policies that affect the level of homophily $h$ in the society.
Any increase in segregation between vaxxers and anti--vaxxers may induce the disease to become endemic. Additionally, a larger $h$, if there is a temporary outbreak, will slow down the recovery time, and in some cases (i.e.~when the outbreak does not start only among vaxxers), it may increase the cumulative infection caused by the disease.


\section{Vaccination choices}
\label{sec:vaccination}

Vaccination rates adjust as homophily varies, because homophily changes the risk perceived by agents. However, under Assumption \ref{corner} of extreme anti-vaxxers, the only relevant variation is in the \emph{vaxxer} group. But, as $h$ increases, the group with more vaccinated people (the vaxxers) is more protected against infection, so the perceived risk $\sigma^v$ decreases, and as a result, a smaller fraction of vaxxers is vaccinated: $x^v$ is decreasing in $h$. A smaller fraction of vaccinated agents, in turn, triggers even larger infection levels. This mechanism works in addition to the standard diffusion mechanism discussed in the previous section, so that an increase in homophily increases infection \emph{even more}. We can formalize this in the following proposition.

\begin{prop}
\label{prop:endovax}
If $d>\overline{d}$ and vaccination choices are endogenous, the cumulative infection is increasing in homophily $h$. Moreover, it is increasing \textit{more} than if vaccination rates were exogenous. 

\end{prop}

The proof follows immediately from the total derivative:
\[
\frac{\dd CI}{\dd h}=\frac{\partial CI }{\partial h}+\frac{\partial CI }{\partial x^v}\frac{\dd x^v}{\dd h},
\]
and observing that Assumption \ref{vax1} implies:
\[
\frac{\dd x^v}{\dd h}=-\frac{\frac{\partial \sigma^v}{\partial h}}{1-\frac{\partial \sigma^v}{\partial x^v}}<0,
\]
so that:
\[
\frac{\dd CI}{\dd h}=\underbrace{\frac{\partial CI }{\partial h}}_{>0}+\underbrace{\frac{\partial CI }{\partial x^v}\frac{\dd x^v}{\dd h}}_{>0}.
\]

Indeed, the derivative is larger than the one for the case of exogenous vaccination choices, as it can be evinced by the fact that both the addends in the expression above are positive.

Another apparently paradoxical phenomenon that is the consequence of endogenous vaccination choices is that not only homophily can be detrimental to total cumulative infection, but also to infection \emph{among vaxxers alone}. We show this under \textit{risk proportional to non-vaccinated} ($NV-$risk, in the terminology of Section \ref{examples}). The mechanism works through the fact that, despite $CI^v$ being decreasing in $h$, as explained above, the vaccination rate among vaxxers $x^v$ is decreasing when $h$ increases, because the perceived risk is smaller. This creates a counterbalancing effect, and if risk is sufficiently high (as parameterized by $k$ in the $NV-$risk case), the effect is strong enough to make $CI^v$ increasing.

\begin{prop}
\label{vax-infection}

Under $NV-risk$, and Assumption \ref{corner}, there exists a $\underline{k}$ such that, if $k>\underline{k}$, $CI^v$ is increasing in homophily $h$.

\end{prop}

So an increase in $h$ cannot be considered unanimously beneficial neither from a planner concerned with total infection, nor from a planner concerned with just infection among vaxxers.

\bigskip

What happens if anti--vaxxers are not too extreme, that is $d<\overline{d}$? This introduces a new mechanism, because by Assumption \ref{vax1}, as homophily increases, anti--vaxxers perceive more infection in their neighborhood, and so increase their equilibrium vaccination rate. This creates a competing effect, and the balance of the two is a priori unclear. In Appendix \ref{app:interior} we show that in the two parametric cases introduced in Section \ref{examples}, the mechanism carries through also if $d<\overline{d}$, at least for a small level of homophily.

\section{Endogenous groups}

\label{sec:groups}

In this section the additional trade-offs generated by the cultural dynamics force us to use a parametric form for $\sigma$. We show the results under the two parametric forms introduced in Section \ref{examples}

We start by showing that only the case in which $\alpha<0$ is of some interest for the analysis, because in the other cases the population become all of one type, with unique stable steady state either $q^*=0$ or $q^*=1$. So, $\alpha<0$ characterizes the conditions under which there exists an interior fraction of anti--vaxxers in the population.

\begin{prop}
\label{existence}

Under Assumptions \ref{corner} and \ref{simple_q}, endogenous vaccination choices, and under both $NV-$risk and $CI-$risk: 
\begin{enumerate}[i)]
\item if $\alpha \geq 0$ there are no interior stable steady states of the dynamics for $q$;

\item if $\alpha<0$, there exists a threshold $d_q$ 
such that if $d>d_q$ there exists a unique stable steady state of the cultural dynamics $q^*\in (0,1)$.
\end{enumerate}

\end{prop}

Again, the proof of this result is obtained with standard methods, applying the implicit function theorem to the condition from Assumption \ref{simple_q}.

The reason for the condition above on $d$ is that if $\Delta U^a=0$, then anti--vaxxers exert no effort, and the only steady states will be with $q=0$. This happens if, for example, the bias $d$ is very high, or homophily is very high, so that the increased infection risk from being an anti-vaxxer (the blue area in Figure \ref{fig:payoffs}) is so large that no one wants to be an anti-vaxxer. This is of course an uninteresting case, so from now on we are going to assume the following:
\begin{assumption}[\textbf{Interiority conditions}]
\label{interior}
$d >\max\{d_q,\underline{d}\}.$

\end{assumption}
It is clear that Assumption \ref{interior} implies Assumption \ref{corner}.
Note also that under the interiority condition, $(x^v)^* \in (0,1)$.

\subsection{Impact of homophily}

In this section we explore what is the impact of homophily on infection once we take into account the adjustment in the fraction of anti--vaxxers. 

It is not possible anymore to deduce this behavior only from Assumption \ref{vax1}, because there are competing forces. The steady state level of $q$ depends on the balance of socialization payoffs, and  socialization payoffs are decreasing with the level of infection in the own group and increasing in the level of infection in the other group. A variation in $q$ raises infection in both areas, so it is not possible to deduce the direction of the effect without any reference to the specific form $\sigma$ takes. For these reason in this section we are giving result for the two possible instances of risk evaluations $\sigma$ detailed in Section \ref{examples}. 

The key observation is in the next proposition: homophily increases risk for anti--vaxxers, so in equilibrium decreases their number.

\begin{prop}
\label{q_decreasing}

Under endogenous groups and Assumptions \ref{simple_q} and \ref{interior}, under both $NV-$risk, $CI-$risk, and endogenous vaccination choices and groups, the share $q$ of anti--vaxxers  is decreasing in homophily $h$. 

\end{prop}

To understand this result, we need to understand the effects of $h$ and $q$ on socialization payoffs.
The direct effect of homophily on the socialization payoffs under our assumptions is unambiguous: since homophily increases risk for anti--vaxxers and decreases it for vaxxers, it follows that an increase in $h$ makes the socialization payoff larger for vaxxers and smaller for anti--vaxxers.
The impact of $q$ is a priori ambiguous, since it increases the risk for both groups, hence the need for specifying the functional form of $\sigma$. In Appendix \ref{subs} we give a more detailed account of the behavior of socialization payoffs as a function of $q$, that is the degree of \emph{cultural substitution} displayed in the model.

Now we have the elements to understand the mechanics of Proposition \ref{q_decreasing}. The intuition for the result is as follows: homophily increases the effort of vaxxers and decreases the effort of anti--vaxxers. Now the cultural substitution effect tends to move effort in favor of vaxxers as $q$ increases. If this is the dominant effect, then as $h$ increases we need a decrease in $q$ to be in the steady state. 
If the condition is violated, we get only a corner solution in which anti-vaxxers disappear.

Now that we have all the elements in place, we can ask what is the global effect of homophily, through the cultural channel, the adjustment of vaccination rates, and the disease dynamics. In addition to the direct effects discussed in the previous paragraphs, the direct effect of homophily on group size has to be taken into account. A larger fraction of anti--vaxxers increases infection. In turn, the size of the anti-vaxxer group increases both cumulative infection and the number of vaccinated, and the two variations have countervailing effects.

\begin{prop}
\label{prop:alpha}

Under endogenous groups and Assumptions \ref{simple_q} and \ref{interior}, under both $NV-$risk, $CI-$risk, and endogenous vaccination choices and groups, if $|\alpha|$ is sufficiently large, cumulative infection is increasing in homophily $h$.
\end{prop}

If $\alpha$ is \emph{large} in magnitude, then the society is rigid in its opinions, and the effects are qualitatively the same that we would have if types and vaccination choices were fixed (Proposition \ref{prop:exogenoush}). If instead $\alpha$ is \emph{small} in magnitude, then the reaction of $q^*$ to a change in $h$ is large, and this might revert the effect: cumulative infection might then be decreasing in homophily. In this respect, how agents are subjected to social influence can revert the effects of a variation in homophily.
 Figure \ref{example_prop4} shows this effect for two values of $\alpha<0$. These are also compared with what would happen, with the same parameters, under the assumptions of Proposition \ref{prop:exogenoush} (all choices are exogenous) and Proposition \ref{propCIend} (only vaccination choices are endogenous, but groups are fixed). The figure shows that, only when $\alpha$ is negative and small in absolute value, the cumulative infection decreases in homophily. In all the other cases, an raise in homophily can increase the cumulative infection at various degrees.
\begin{figure}[h!]
		\includegraphics[width=\textwidth]{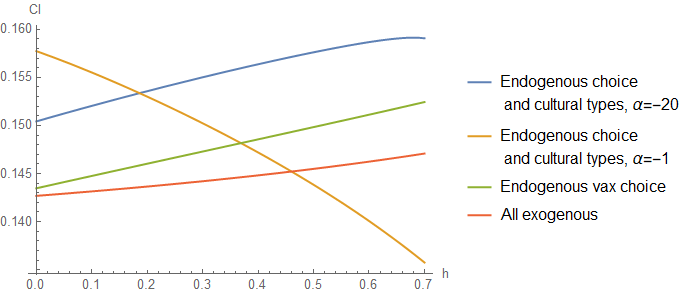}
	\caption{Cumulative infection in the three models under $NV-$risk. Whenever exogenous, $q$, and $x^v$ are set using the mean values in the range. The other parameters are set at $k=1$, $d=1$, $\mu=2$, $\rho^a_0=\rho^v_0=0.2$.}
	\label{example_prop4}
	\end{figure}

The intuition for the different marginal effects of $h$ on cumulative infection seems to lie on the marginal effects on the speed of the dynamics, via the first eigenvalue (see Proposition \ref{convtime}), as Figure \ref{endogenousq:CIcorner} illustrates: the cases in which cumulative infection increases with $h$ are those in which the leading eigenvalue is decreasing in magnitude, and vice versa.

\begin{figure}
  \centering
\includegraphics[width=0.45\textwidth]{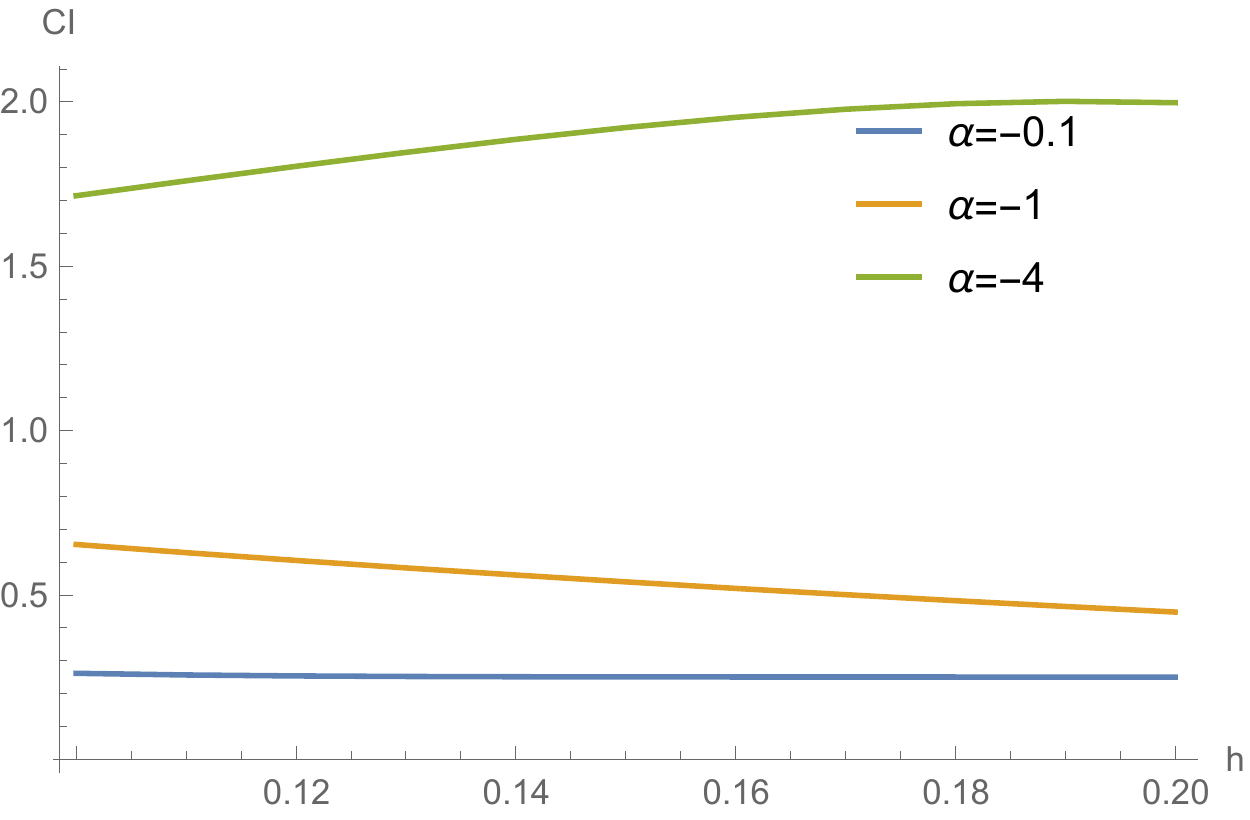}
\includegraphics[width=0.45\textwidth]{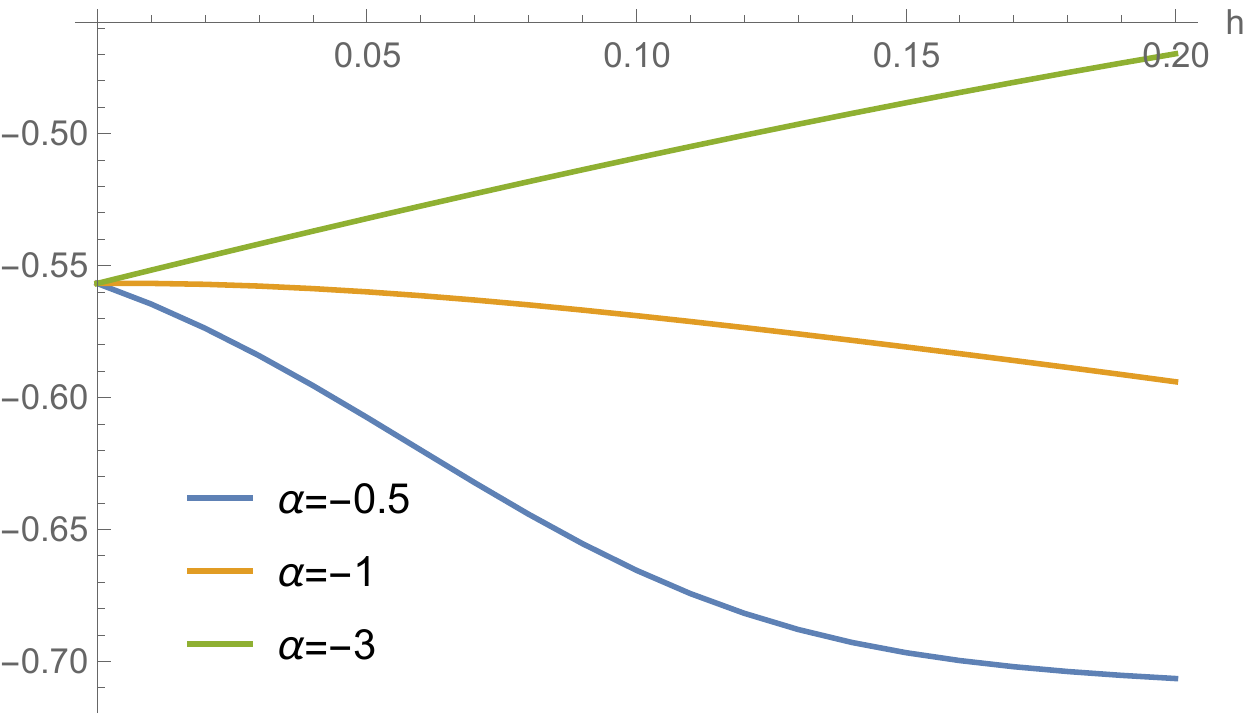}
\caption{{\bf Left panel}: Cumulative infection as function of homophily if $d> \underline{d}$ ($(x^a)^*=0$), in the proportional risk of infection case. The other parameters are set at $k=1$, $d=0.6$, $\mu=0.7$, $\rho^a_0=\rho^v_0=0.1$.
{\bf Right panel}: corresponding leading eigenvalue of the dynamical system as a function of $h$. 
}
\label{endogenousq:CIcorner}
\end{figure}





	
	
	
	
	
	
	


\section{Conclusion}
\label{sec:conclusion}

The problem of vaccine skepticism is a complex one, that requires analysis from multiple angles, e.g., psychological, medical, and social. The results of this paper might be relevant for a policy maker interested in minimizing infection in a world with vaxxers and anti--vaxxers, having available a policy inducing some degree of segregation, or homophily, $h$. 
The key observation is that reducing contact with anti--vaxxers may be counterproductive both from the perspective of vaxxers and of the society as a whole because it slows down the dynamics of the disease to its steady state, if there is an outbreak.
Homophily may actually increase the duration of the outbreaks and, depending on the time preferences of the planner, this might crucially change the impact of the policy. Further, if belonging to the vaxxers or anti--vaxxers group is endogenous, the intensity of cultural substitution is key in determining the impact of the policy. Our results suggest that the study of policy responses to the spread of vaccine-hesitant sentiment would benefit from trying to pin down more precisely the intensity of these mechanisms.
	
	\clearpage
	
\appendix

\section*{Appendices}  

\label{appendix}

\section{A simple model of intragenerational cultural transmission}
\label{app:cultural}

In this appendix we illustrate how equation \eqref{simple_q} with $\alpha=-1$ can arise from a simple adaptation of the \cite{bisin2001economics} model to an intragenerational context.

At each time period, each agent meets another agent selected randomly. When they meet, they are assigned two roles: the \emph{influencer} and the \emph{target}. The incentive for the influencer is based only on \emph{other--regarding} preferences, for two reasons: it is consistent with some survey evidence (\citealt{kumpel2015news}, \citealt{walsh2004makes}), and in this economy every agent has negligible impact on the spread of the disease, so socialization effort cannot be driven by the desire to minimize the probability of infection, or similar motivations. The timing of the model is as follows.
\begin{itemize}
	
	\item Before the matching, agents choose a \textbf{proselitism effort} level $\tau_t^a$, $\tau_t^v$;
	
	\item When 2 agents meet, if they share the same cultural trait nothing happens. Otherwise, one is selected at random with probability $\frac{1}{2}$ to exert the effort and try to have the other change cultural trait.
	
\end{itemize}

The fraction of cultural types evolves according to:
\begin{equation}
q^a_{t+1}=q_t^aP^{aa}_t+(1-q^a_t)P^{va}_t,
\end{equation}
where the transition rate $P^{aa}_t$ is the probabilities that an agent $a$ is matched with another agent who, next period, results to be of type $a$ and $P^{va}_t$ is the probabilities that an agent $v$ is matched with another agent who, next period, results to be of type $a$. These probabilities are determined by efforts according to the following rules:
\begin{align}
P^{aa}_t & =\tilde{q}_t^a+(1-\tilde{q}_t^a)\frac{1}{2}+(1-\tilde{q}_t^a)\frac{1}{2}(1-\tau^v_t),\\
P^{va}_t & =\frac{1}{2}(1-\tilde{q}_t^a)\tau_t^v,
\end{align}

($P^{vv}_t$ and $P^{av}_t$ are defined similarly) which yield the following discrete time dynamics:
\begin{equation}
\Delta q^a_{t}=q_t^a(1-q_t^a)(1-h)\Delta \tau_t ,   
\end{equation}
where $\Delta \tau_t:=\tau_t^a-\tau_t^v$.

Effort has a psychological cost, which, as in \cite{bisin2001economics}, we assume quadratic. Hence, agents at the beginning of each period (before the matching happens) solve the following problem: 
\begin{equation}
\max_{\tau_t^a}\underbrace{-\frac{(\tau^a_t)^2}{2}}_{\text{cost of effort}}+\underbrace{q^a_tU_t^{aa}+ (1-q^a_t)\frac{1}{2}(\tau^{a}_t U_t^{aa}+(1-\tau_t^{a})U^{av}_t)}_{\text{expected social payoff}},
\end{equation}
which yields as a solution:
 \begin{align}
 \tau_t^a&=(1-q_t^a)\underbrace{(U_t^{aa}-U_t^{av})}_{\text{"cultural intolerance"}},\\
  \tau_t^v&=(1-q_t^v)(U_t^{vv}-U_t^{va}).
 \end{align}
 
 Hence, the dynamics implied by our assumptions is:
 \begin{equation}
\Delta q^a_{t}=q_t^a(1-q_t^a)((1-q_t^a)\Delta U^a-q_t^v\Delta U^v).    
\end{equation}
The steady state of this dynamics is determined by the equation:
 \begin{equation}
(1-q_t^a)\Delta U^a=q_t^v\Delta U^v ,   
\end{equation}
which is precisely the steady state implied by \eqref{simple_q} when $\alpha=-1$.

\section{Mild anti-vaxxers}
\label{app:interior}

In this section we explore some generalizations of the results of the main text to the case in which $d<\underline{d}$, i.e., the bias of the $a$ group is not so large so that some \enquote{anti--vaxxers} do vaccinate in equilibrium: $(x^a)^*>0$. Hence, this case can be taken as a description of \emph{vaccine hesitancy} rather than total refusal. 


In this case, solving we obtain:
\begin{eqnarray}
x^a & = & 1 - \frac{1 + d q^a}{1 + k} - \frac{ d (1 - q^a)}{1 + h k}, \nonumber \\
x^v & = & 1 - \frac{1 + d q^a}{1 + k} + \frac{ d q^a }{1 + h k} \ \ .
\end{eqnarray}
	


This is true provided $d < \min \left\{ \frac{1}{k^2} , \frac{k}{k+1} \right\}$. We use this interiority condition as a maintained assumption for the remainder of this section.


First of all, we note that (i) $x^v>x^a$  - since vaxxers perceive a lower vaccination costs than anti-vaxxers; (ii) $x^a$ is increasing  in $h$ whereas $x^v$ is decreasing in $h$ - since a higher homophily makes vaxxers more in contact with agents who are less susceptible than anti-vaxxers and, as a consequence, $(x^v-x^a)$ is decreasing in $h$; (iii) $x^a$ and $x^v$ are increasing in $q^a$ -  since the higher the share of anti-vaxxers, the more agents are in touch with other subjects at risk of infection; (iv) the total number of vaccinated people is $q^a x^a + (1-q^a)x^v =\frac{k-dq^a}{1+k}$, it is independent of $h$, but decreasing in $q^a$ - this is due to a Simpson paradoxical effect: both groups vaccinate more, but since anti-vaxxers increase, in the aggregate vaccination decreases.

In the case of \textit{proportional infection risk} it is possible to characterize analytically the behavior of the cumulative infection, as in the following proposition.

\begin{prop}
\label{propCIend}
Under $NV-$risk and endogenous vaccination choices, if $d<\overline{d}$, then
$CI$ is increasing in $h$, though \emph{less} than in the case in which vaccination rates are exogenous.

\end{prop}

So, in this case the adjustment of vaccination rates \emph{mitigates} the perverse effect of homophily, though not in a way strong enough to offset it completely.

\bigskip

In case of risk proportional to cumulative infection, we can characterize the behavior analytically for $h$ close to 0.

\begin{prop}
\label{propCIend1}

Under $CI-$risk and endogenous vaccination choices, if $d<\overline{d}$, there exists a $\underline{h}$ such that for $h<\underline{h}$ $CI$ is increasing in $h$, and is \emph{more increasing} than in the case in which vaccination rates are exogenous.

\end{prop}

In other words, in this case the adjustment of vaccination rates \emph{exacerbates} the perverse effect of homophily.

\subsection{Endogenous groups}
\label{app:endgroup}

Integrating, we find that the \textit{socialization payoffs} in this case are:
\begin{align}
\Delta U^a&=	\frac{1}{2}(x_v-x_a)^2- (d- \left(x_v-x_a\right))\left(1-x_v\right),\\
\Delta U^v&=	\frac{1}{2}(x_v-x_a)^2+ (d- \left(x_v-x_a\right))\left(1-x_a\right).
\end{align}

To understand the socialization payoffs in this case consider Figure \ref{fig:payoffs2}. The black line is the disutility of agents in group $a$, as a function of the cost $c$, \emph{as perceived by agents in group $a$}. The shape of this line mirrors the fact that an agent in group $a$ undertakes vaccination only if her costs are in the $[0,k\sigma^a-d]$ interval, in which $a$ agents incur in a  disutility $c+d$. If $c>k\sigma^a-d$, $a$ agents do not vaccinate, and the disutility is the risk of infection, which is $k\sigma^a$. The grey area below this curve is then $U^{a a}$. Consider now the red line. This represents the disutility of agents in group $v$ \emph{as perceived by agents in group $a$}. In particular, agents in group $v$ have a different perception of costs with respect to agents in group $a$, and so take different choices. In particular, they vaccinate in the  $[0,k\sigma^v]$  interval, while if they are in the $[k\sigma^v,1]$ interval they do not vaccinate and incur a risk of infection. Note, however, that this is the evaluation from the perspective of agents in group $a$, and thus the cost of vaccination is $c+d$ instead of $c$. Hence $ U^{a v}$ is the area below the red curve. The difference $\Delta U^{a}$ is given by the red area minus the blue area. $ U^{v v}$ and $ U^{v v}$ are computed accordingly. 

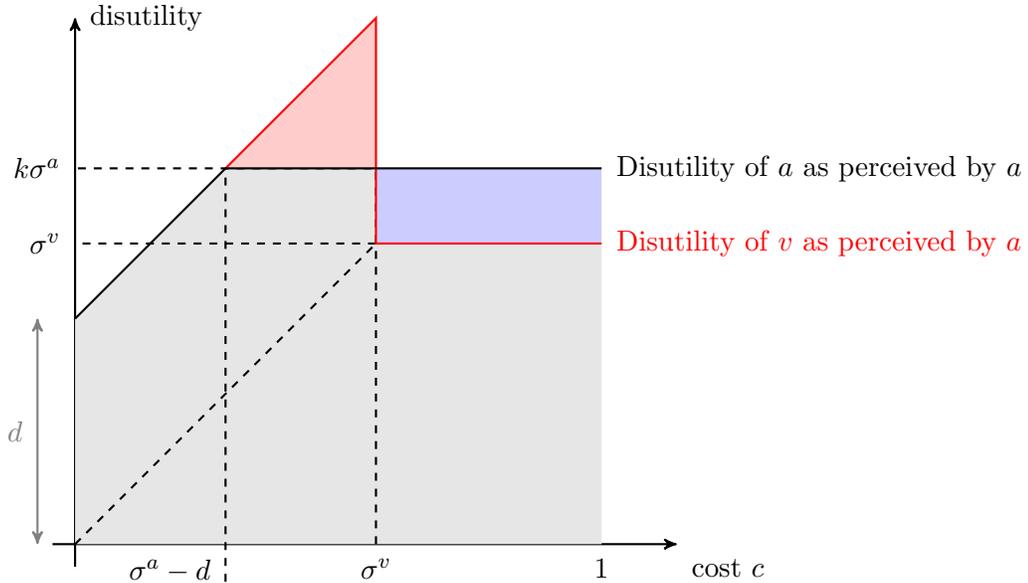
\begin{figure}[h]
    \centering

\begin{tikzpicture}[
thick,
>=stealth',
dot/.style = {
draw,
fill = white,
circle,
inner sep = 0pt,
minimum size = 4pt
}
]
\coordinate (O) at (0,0);
\draw[->] (-0.3,0) -- (8,0) coordinate[label = {below right:cost $c$}] (xmax);
\draw[->] (0,-0.3) -- (0,7) coordinate[label = {right: disutility}] (ymax);

\fill[gray!20] (0,0) -- (0,3) -- (2,5) -- (7,5) -- (7,0)  ;
\draw[dashed] (0,0) -- (5,5);
\draw[gray, <->] (-0.5,0) -- node[left] {$d$} (-0.5,3);

\draw [dashed] (2,-0.5) -- (2,5);

\draw[dashed] (4,0)--(4,5)--(0,5) node[left] {$k\sigma^a$};
\node[below] at (7,0) {$1$} ;   
\node[below left] at (2,0) {$\sigma^a-d$} ;   
\node[below] at (4,0) {$\sigma^v$} ;   

\draw[dashed] (4,4)--(0,4) node[left] {$\sigma^v$};	
\fill[red!20] (2,5) -- (4,7) --  (4,5) ;
\fill[blue!20] (4,4) -- (4,5) --  (7,5) --(7,4) ;   
\draw[red] (2,5) -- (4,7)-- (4,5)--  (4,4) --(7,4) node[right]{Disutility of $v$ as perceived by $a$};
\draw (0,3) -- (2,5) --  (7,5) node[right]{Disutility of $a$ as perceived by $a$}; 

\end{tikzpicture}

    \caption{Composition of $\Delta U^a$. The graph represents the disutility incurred by an individual as a function of its cost $c$. $\Delta U^a$ is the red area minus the blue area. }
    \label{fig:payoffs2}
\end{figure}

First, we again prove an existence result.

\begin{prop}
\label{prop:existence_appendix}
If $\alpha<0$, there exists an interior steady state $q^*$  ,of the cultural dynamics provided $d$ is large enough, that is: $\frac{2 h k (h k+1)}{k+1}<d$.

\end{prop}

The condition on $d$ is the condition under which anti-vaxxers exert enough effort and $\Delta U^a$ is always positive. Otherwise, we get only a steady state with $q=0$. 
For this to be compatible with $x^a$ and $x^v$ being interior, we need $d<\min\{\frac{1}{k},\frac{1}{k(1+k)}\}$, hence we need also $2h(1+hk)<\frac{1}{k^2}$. So, in addition to $k$ high enough we also need $h$ small enough. We are going to assume this condition in the following. Figure \ref{figura_vele} shows the regions in the $h$--$d$ plane for which the interiority conditions are satisfied, depending on the value of $k$.

\begin{figure}[h!]
\includegraphics[width=0.6\textwidth]{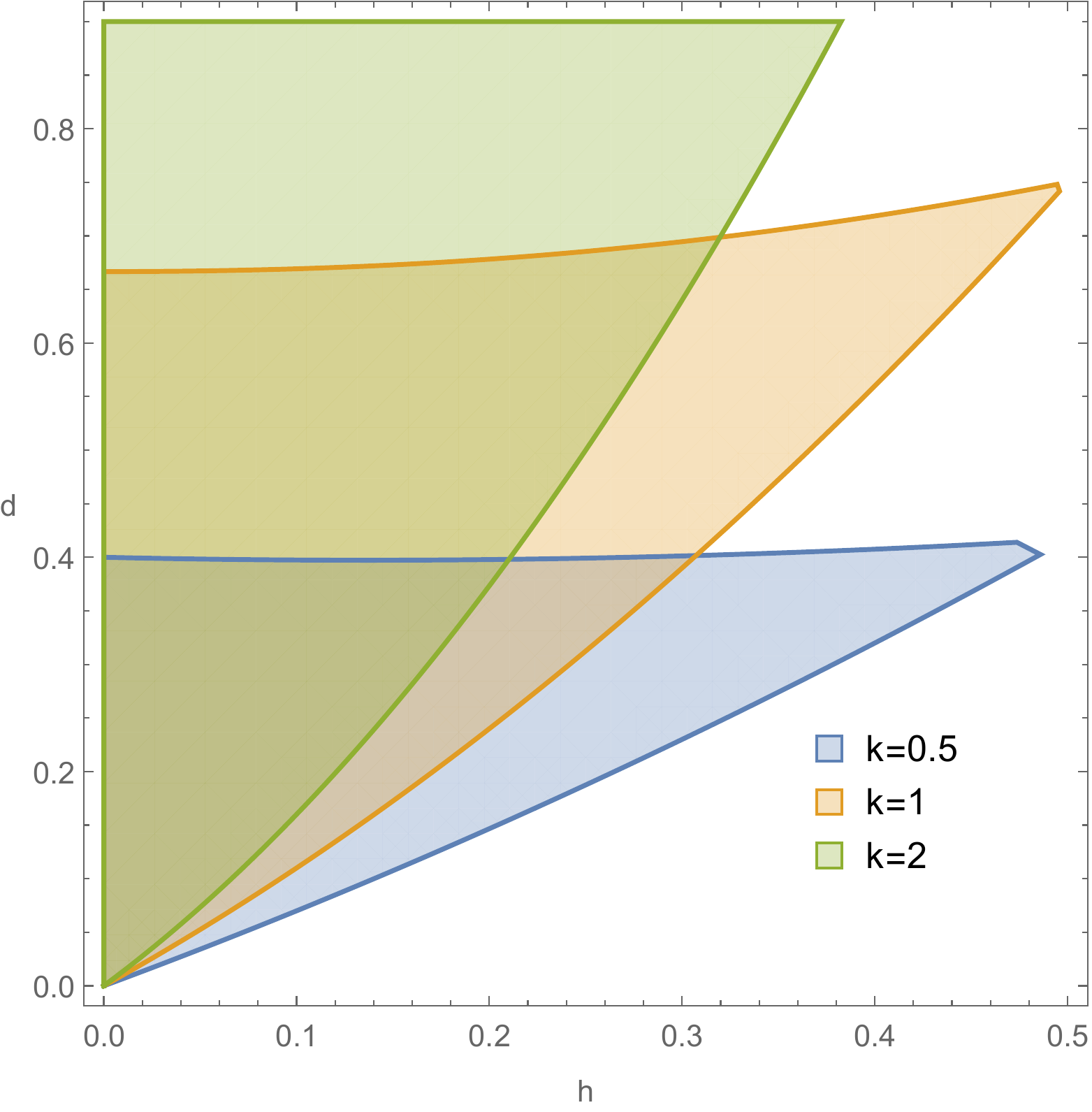}
\caption{Region of parameters where all endogenous variables are interior. $\mu$ is fixed to 1.}
\label{figura_vele}
\end{figure}

We are now interested in the effect of an increase in homophily on $q^*$. Figure \ref{q_with_h} shows, on the basis of numerical examples with $\alpha=-\frac{1}{2}$, $\alpha=-1$, and $\alpha=-3$, that homophily has a negative effect on $q^*$ and that this result seems to extend to any $\alpha<0$. Analytical tractability, however, is obtained only for values of $h$ that are \emph{small}, as would be the effect of a policy that limits contacts between vaxxers and anti-vaxxers only in a few of the daily activities (e.g.~only in schools).

\begin{figure}
  \centering
	\includegraphics[width=8cm]{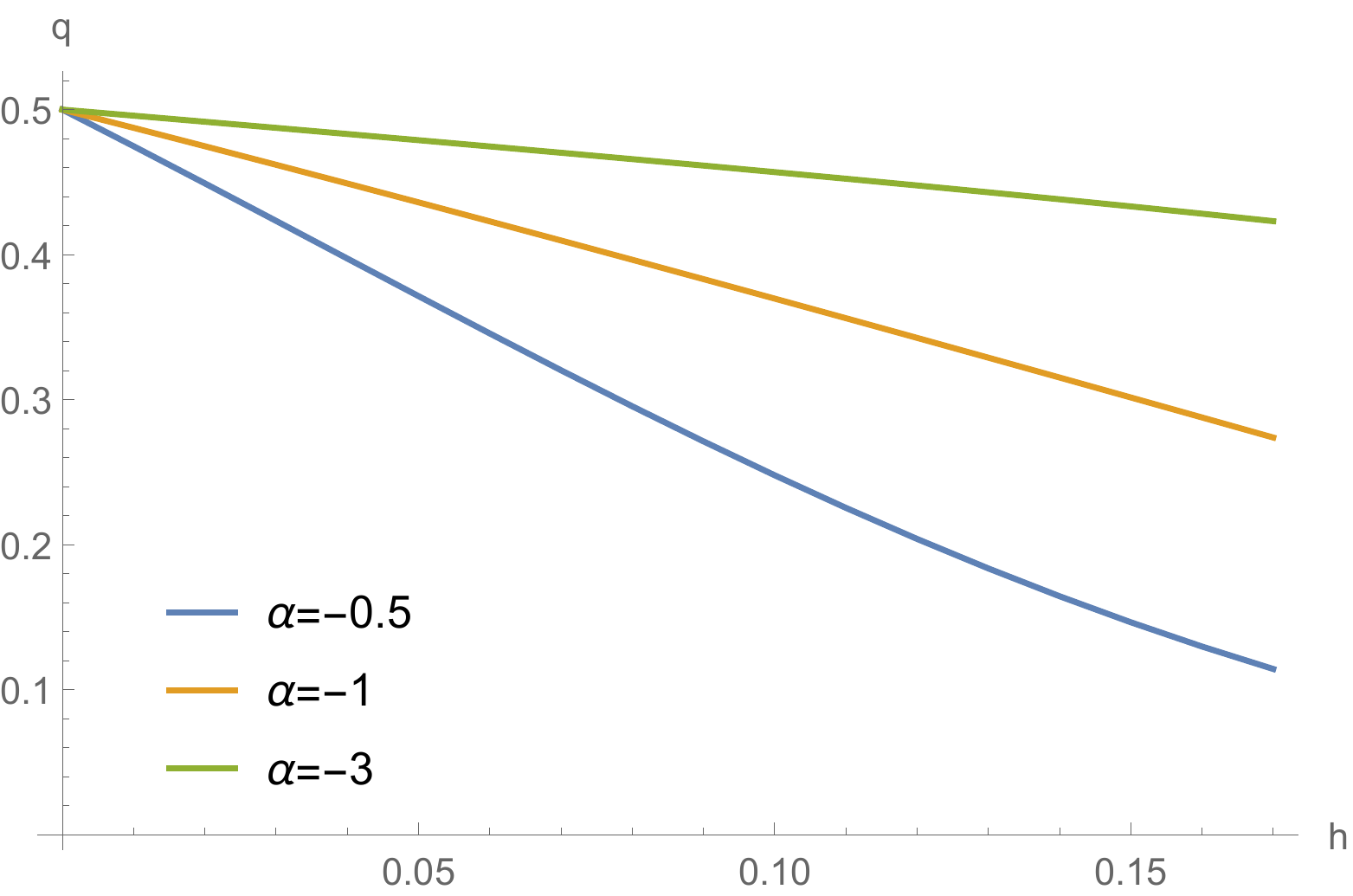}
	\caption{$q$ as a function of $h$. $d=0.5$, $k=1$, $\mu=1$. The range of $h$ is restricted as prescribed by the interiority conditions.}
\label{q_with_h}
\end{figure}
We can actually prove it analytically for \emph{small} values of $h$.
\begin{prop}
\label{derivative_qh}
Under the interiority conditions, and if $\alpha<0$
there exists a $\overline{h}$ such that, for $h<\overline{h}$, the unique interior steady state $q^*$, which is also stable, is decreasing in $h$. 
\end{prop}
The intuition is that a larger $h$ magnifies the negative effects of being anti-vaxxers in terms of infection, relatively to vaxxers. This is internalized in the cultural dynamics, \textit{via} the $\Delta U$s. This long run effect of $h$ on anti-vaxxers share is one of the few positive effects of segregating policies.\\
As we have done in the preliminary model with exogenous choices, we can analyze the effects of homophily on the cumulative infection, when the initial perturbation is symmetric across both groups (see Proposition \ref{prop:exogenoush}, summarized in the third column of Table \ref{marginal_effects}).
We find that the effects depend on the magnitude of $\alpha$, the parameter regulating how agents are rigid/prone towards social influence.
\begin{prop}
\label{endogenousq}
Consider the model with endogenous $q$, $\alpha<0$, and interiority conditions.
Consider an outbreak affecting both groups symmetrically, starting from the unique stable steady state and $h=0$. Then, there exists a threshold $\overline{\alpha}$ such that: 
\begin{itemize}
    \item if $\alpha<\overline{\alpha}$, CI is increasing in $h$ $\left( \left. \frac{\dd CI}{\dd h} \right|_{ h=0}>0 \right)$;
    \item  if $\alpha>\overline{\alpha}$, CI is decreasing in $h$ $\left( \left. \frac{\dd CI}{\dd h} \right|_{ h=0}<0 \right)$.
\end{itemize}

\end{prop}

Numerical simulations reveal a picture very similar to the one described in the main text. Specifically, the magnitude of $\alpha$ is crucial to determine the effect of an increase in homophily, as illustrated in Figure \ref{endogenousq:CIinterior}.
Again, the main mechanism through which homophily acts is via the increased length of the outbreak, as measured by the leading eigenvalue, as shown in the same figure. Figure \ref{example_prop4mild} compares the behavior of cumulative infection in the three different models.

\begin{figure}
  \centering
\includegraphics[width=0.45\textwidth]{CI_corner.pdf}
\includegraphics[width=0.45\textwidth]{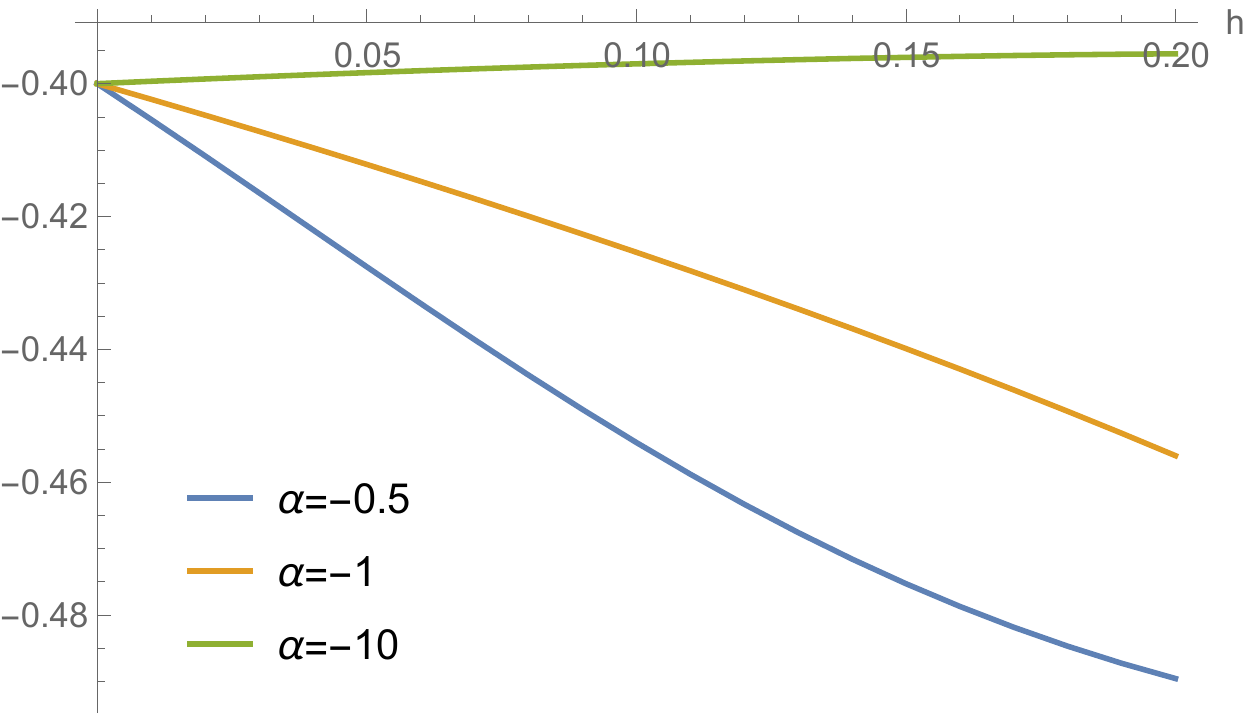}
\caption{{\bf Left panel}: Cumulative infection as function of homophily in the interior equilibrium. The other parameters are set at $k=2$, $d=0.5$, $\mu=1$, $\rho^a_0=\rho^v_0=0.1$.
{\bf Right panel}: corresponding leading eigenvalue of dynamical system as a function of $h$. 
}
\label{endogenousq:CIinterior}
\end{figure}

\begin{figure}[h!]
		\includegraphics[width=\textwidth]{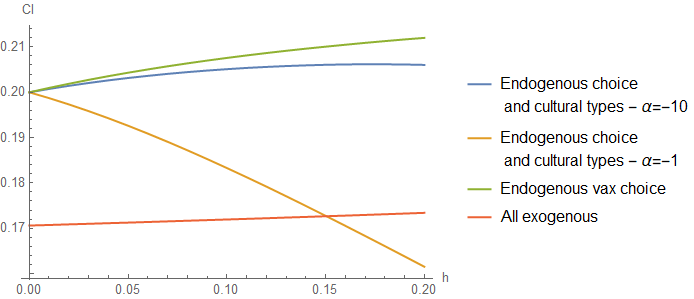}
	\caption{Cumulative infection in the three models. Whenever exogenous, $q$, and $x^v$ are set using the median value of $h=0.1$. The other parameters are set at $k=2$, $d=0.5$, $\mu=1$, $\rho^a_0=\rho^v_0=0.1$.}
	\label{example_prop4mild}
	\end{figure}



\clearpage

\section{Other extensions}
\label{sec:extensions}

\subsection{Asymmetric initial infections}

So far, we assumed that the initial seeds of the infection are selected at random, independently of group identity. If the independence hypothesis is relaxed, we might have a different fraction of initially infected in the two groups.
In this section we illustrate the role of the initial conditions in determining the behavior of infection, generalizing the results in Section \ref{sec:mechanical}. 

To illustrate the mechanics that regulates the share of agents that get infected during the outbreak, let us consider three different types of initial conditions: The epidemic starts (i) among vaxxers ($\rho^v_0>0$ and $\rho^a_0=0$), (ii) among anti-vaxxers ($\rho^v_0=0$ and $\rho^a_0>0$), and (iii) in both groups symmetrically ($\rho^v_0=\rho^a_0>0$, the case explored in the main text). In the following sections we are going to stick to the case in which $\rho^a=\rho^v=\rho_0$.

The first result we present generalizes Proposition \ref{prop:exogenoush}.
\begin{prop}[Who is better off?]
\label{prop_whobetter}

The cumulative number of infected agents is such that $CI^a \ge CI^v$ if and only if:
\begin{equation}
    \rho^v_0(1-x^a)-\rho^a_0(1-x^v)+\mu(\rho^a_0-\rho^v_0)\ge 0
    \label{ineq}
\end{equation}

\end{prop}

The result simply follows from comparing the explicit expressions for $CI^a$ and $CI^v$ (we derive it in Lemma \ref{cuminf} in the Appendix \ref{app:proofs}). Inequality \eqref{ineq} underlines the roles of the parameters in determining the welfare of the groups. The left--hand side is increasing in $x^v$ and decreasing in $x^a$: the gap in vaccinations tends to penalize the less vaccinated group. Since the cumulative infection is an intertemporal measure, the initial conditions also concur in determining which group is better off: the difference is increasing in $\rho^a_0$ and decreasing in $\rho^v_0$.\footnote{Because the stability assumptions imply $-1+x^v+\mu>0$ and $-1+x^a+\mu>0$.} $\mu$ regulates the importance of this effect in the discrepancy of initial conditions: the larger $\mu$, the shorter the epidemic, the larger the importance of the initial conditions. In particular, we have:
\begin{enumerate}[i)]
    \item if the outbreak starts among vaxxers, vaxxers have a larger cumulative infection;
    \item if the outbreak starts among anti--vaxxers, anti--vaxxers have a larger cumulative infection;
\item if the outbreak starts symmetrically in both groups, the group with less vaccinated (anti--vaxxers, under our assumptions) has the largest cumulative infection.
\end{enumerate}

In particular, the evaluation of what group is better off in terms of infections is \emph{independent} of homophily. However, the levels of contagion do depend on homophily, as the following result shows. It is obtained applying definitions from \eqref{cumulative} and taking derivatives.

\begin{prop}[Effect of $h$ and $q^a$]
\label{prop:exogenoush2}
\textcolor{black}{
\begin{enumerate}[a)]
\item $CI$ and $CI^a$ are increasing (decreasing) in $h$ if and only if $CI^a>CI^v$ ($CI^a<CI^v$); $CI^v$ is decreasing (increasing) in $h$ if and only if $CI^a>CI^v$;
\item $CI$, $CI^a$ and $CI^v$ are increasing (decreasing) in $q$ if and only if $CI^a>CI^v$ ($CI^a<CI^v$).
\end{enumerate}}

In particular, the marginal effects of $h$ and $q^a$ for different outbreak types are those reported in Table \ref{marginal_effects}. 


\end{prop}

\begin{table}[h]
    \centering
    \begin{tabular}{c|c|c|c|}
     &  
 \multicolumn{3}{|c|}{\mbox{If the outbreak is among$\cdots$}} \\
 \hline  
         & \mbox{vaxxers} & \mbox{anti--vaxxers} & \mbox{symmetric: $\rho_0^a=\rho_0^v$} \\
         \hline
         \rule{0pt}{3ex} \mbox{the effect} & $\frac{\partial CI^a}{\partial h}<0$, $\frac{\partial CI^v}{\partial h}>0$, 
         & $\frac{\partial CI^a}{\partial h}>0$, $\frac{\partial CI^v}{\partial h}<0$,
         & $\frac{\partial CI^a}{\partial h}>0$, $\frac{\partial CI^v}{\partial h}<0$,  \\
        \mbox{of $h$ is:}  & $\frac{\partial CI}{\partial h}<0$ &  $\frac{\partial CI}{\partial h}>0$ & $\frac{\partial CI}{\partial h}>0$ \\
         \hline
         \rule{0pt}{3ex} \mbox{the effect} & $\frac{\partial CI^a}{\partial q}<0$, $\frac{\partial CI^v}{\partial q}<0$, 
         & $\frac{\partial CI^a}{\partial q}>0$, $\frac{\partial CI^v}{\partial q}>0$,
         & $\frac{\partial CI^a}{\partial q}>0$, $\frac{\partial CI^v}{\partial q}>0$,  \\
         \mbox{of $q^a$ is:} & $\frac{\partial CI}{\partial q}<0$ &  $\frac{\partial CI}{\partial q}>0$ & $\frac{\partial CI}{\partial q}>0$ \\
         \hline
    \end{tabular}
    \caption{Marginal effects of $h$ and $q^a$ on $CI^a$, $CI^v$, and $CI$, when there is an outbreak among vaxxers, anti--vaxxers, or symmetrically in both groups.}
    \label{marginal_effects}
\end{table}

The previous results show how initial conditions and parameters contribute to determining the effect of an increase in $h$. As anticipated in the introduction, if the initial parameters are such that $CI^a=CI^v$, then both the total infection, $CI$, and the group level ones, $CI^a$ and $CI^v$, do not depend on homophily.
If instead, the initial parameters are such that $CI^a \ne CI^v$, then homophily hurts the group with more infected, because it causes the infection to spread to more members of the group and less outside. Table \ref{marginal_effects} helps us understand the behavior in prototypical cases and analyze whether a policy that increases $h$ has the desired effect.

To better understand the mechanics, let us first focus on the effects of homophily (first row of Table \ref{marginal_effects}).
First note that, if the outbreak happens just in one of the two groups, 
homophily protects the group that is not infected ex-ante. 
So, intuitively, the outbreak has the strongest effect in terms of infected agents in the group in which the outbreak has taken place. The effect of homophily on the overall $CI$ is however ambiguous and depends on the initial condition.

Consider first the case in which the outbreak takes place among vaxxers. Then, at the beginning, the infection takes over among the group with the highest vaccination rate, since $x^v>x^a$. The higher the homophily $h$, the more vaxxers interact with each other, and thus the more the infection remains within the group that is more protected against it. For this reason, the higher $h$, the less the $CI$. For the opposite reason, if the outbreak takes place in the anti-vaxxers group, homophily makes infection stay more in the less protected group, and $CI$ increases.

So the crucial message is that a policy having the effect of increasing $h$ cannot be considered unanimously beneficial neither from a planner concerned with total infection, nor from a planner concerned with just infection among vaxxers.





To understand the role of $q^a$ on the $CI$ (second row of Table \ref{marginal_effects}), recall that a higher $q$ means a higher share of agents less protected against the disease. Consider first the case in which the outbreak takes place in the vaxxers group. Then, a higher $q$ means that the number of infected agents, which are in the $v$ group, is lower. Thus, all $CI$ measures are decreasing in $q$. For the opposite reasoning, all $CI$ measures are increasing in $q$ if the outbreak takes place in the anti-vaxxers group. If the outbreak is symmetric, then the two forces mix. However, if $q$ increases, the share of agents who are not protected against the disease increases, and thus $CI$ measures increase.

\subsection{Time preferences}
\label{discounting}

In this section we explore the implications of the degree of impatience of the planner on the evaluation of the impact of homophily. Time preferences can be crucial for the planner. As we have seen, for example, in the Covid19 epidemic, the planner, given a CI, may prefer not to have all infected agents soon because of some capacity constraints of the health system.

 {For example, Figure \ref{duringtime} shows the time evolution of the infection of both groups and the overall society in case of an outbreak among the vaxxers. In this case, since the outbreak starts among the vaxxers, it is among this group that infection is higher initially. In contrast, eventually infection becomes larger among the anti--vaxxers, due to the lower vaccination levels. The effects on cumulative infection depend on how the planner trades off today and tomorrow infections: the more the planner is patient, the more the infection among anti--vaxxers becomes prominent. 

Moreover, since in our setting the impact of segregation policies depends on the relative amount of infected agents in the two groups, as specified in Result \ref{prop:exogenoush}, in our context the time preference is also crucial for the evaluation of the impact of homophily on the total cumulative infection.}


\begin{figure}
  \centering
	\includegraphics[width=8cm]{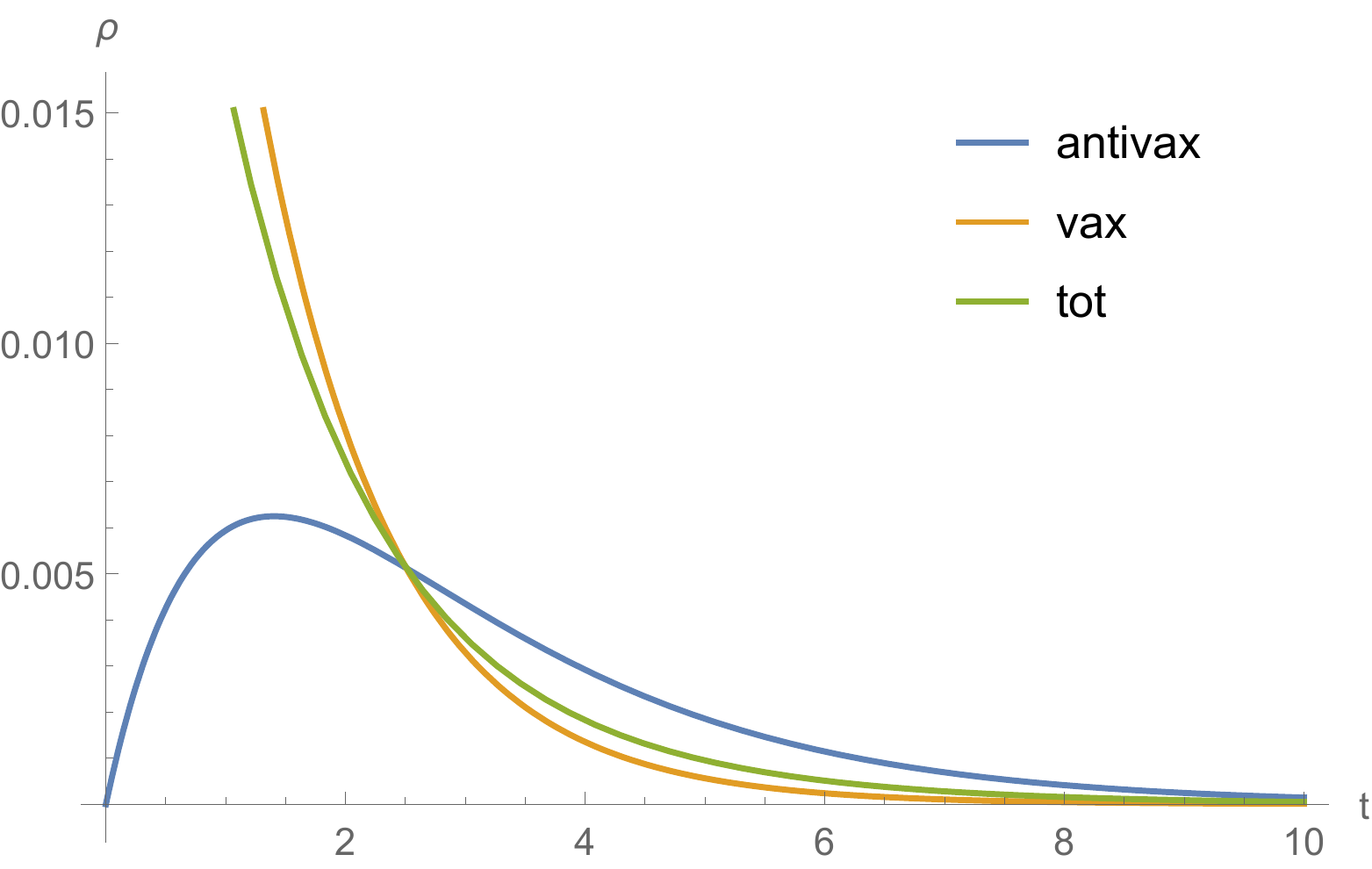}
	\caption{$CI$ as a function of time in case the outbreak starts among vaxxers ($\rho^a_0=0$). Here $\rho^v_0=0.1$, $x_a=0.3$, $x_v=0.9$, $q=0.3$, $h=0.5$, $\mu=1$.  }
\label{duringtime}
\end{figure}

Thus, we first define the \emph{discounted} cumulative infection:
\begin{align}[left=\empheqlbrace]
CI^a&:=\int_0^{\infty}e^{-\beta t}\rho^a(t)d t \nonumber \\
CI^v&:=\int_0^{\infty}e^{-\beta t}\rho^v(t)d t \label{cumulativetime} \\
CI&:=q^aCI^a+(1-q^a)CI^v \nonumber 
\end{align}
where $\beta>0$ is the discount rate. Analytically, things turn out to be very simple, due to the exponential nature of the solutions, as the following observation lays out.
 \begin{prop}
 
 \label{betamu}
Discounted cumulative infections are equivalent to cumulative infections in a model with recovery rate $\mu'=\mu+\beta$.
 
 \end{prop}

This is not too surprising: $\mu$ is a measure of how fast the epidemic dies out, and $\beta$ is a measure of how fast the welfare loss dies out. 
The previous result carries on even when, as we do in the following sections, choices on vaccination and on types are made endogenous.

The impact can be made more precise if we stick to exogenous choices, as it is done below.

\begin{prop}

In the model with discounting:

 $CI^a\ge CI^v$ if and only if    $ -\rho^a_0(1-x^v)+\rho^v_0(1-x^a)+(\mu+\beta)(\rho^a_0-\rho^v_0)\ge 0 $

\end{prop}

The proof is immediate from the previous result and from Proposition \ref{betamu}. In details: 

\begin{enumerate}
    \item An increase in the degree of impatience $\beta$ makes initial conditions more important for the welfare evaluation. For example, without time preferences, we may have that $\rho^a_0<\rho^v_0$ but $CI^a>CI^v$, because the difference in vaccinated agents dominates the difference in the initial outbreak. However, if time preferences are introduced, or $\beta$ gets larger, a planner may evaluate that $CI^a<CI^v$ because she is putting more weight on the earlier moments of the epidemic. 
    
    \item An increase in the degree of impatience $\beta$ can change the impact of homophily,  as illustrated in Figure \ref{beta}. To understand this point, given a population share $q$, there exists a $\beta$ such that homophily does not impact the CI (with time preferences). In this CI, groups get infected at different rates over time. As we change $\beta$, the planner gives more weight to the group getting infected earlier. As we have seen above, homophily plays a role in this process, keeping the infection more into each group. In Figure \ref{beta}, we consider the case in which $q=.3$, so that there are more vaxxers than anti-vaxxers, and vaxxers are also more vaccinated. Thus, the more the planner is impatient, 
    the more she is satisfied by the fact that most agents (vaxxers) are less infected when homophily increases. 
    
\end{enumerate}

\begin{figure}[h]
    \centering
    \includegraphics[width=0.6\textwidth]{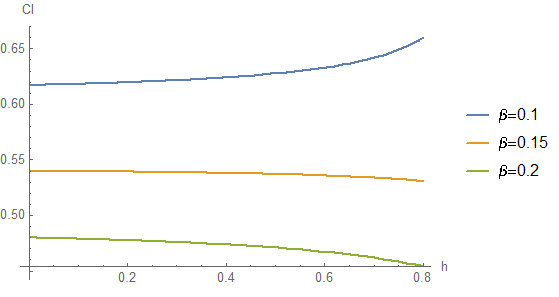}
    \caption{Cumulative infection as a function of homophily for different values of time preference. Here $\mu=0.7$, $x^a=0.2$, $x^v=0.9$, $q=0.3$.}
    \label{beta}
\end{figure}

\clearpage

\section{Proofs}
\label{app:proofs}

\subsection*{Proof of Proposition \ref{SS}}

\begin{proof}

To analyze stability, we need to identify the values of parameters for which the Jacobian matrix of the system is negative definite when calculated in $(0,0)$. The matrix is:
\[\mathbf{J}=
\left(
\begin{array}{cc}
 \left(1-x_a\right) \tilde{q}^a-\mu  & \left(x_a-1\right) \left(\tilde{q}^a-1\right) \\
 \left(x_v-1\right) \left(\tilde{q}^v-1\right) & \left(1-x_v\right) \tilde{q}^v-\mu  \\
\end{array}
\right)\]
We can directly compute the eigenvalues, which are:
\begin{eqnarray}
e_1 & = &  \hat{\mu} - \mu    \nonumber \\
   e_2 & = & \hat{\mu} - \mu -\Delta . \nonumber 
\end{eqnarray}
where $\hat{\mu}:=\frac{1}{2}\left(T+\Delta\right)\in[0,1]$, $T:=\tilde{q}^a(1-x^a)+\tilde{q}^v(1-x^v)$, and $\Delta:=\sqrt{T^2-4h(1-x^a)(1-x^v)}$.	

The eigenvalues are real and distinct because, given $(x+y)^2>4x y$ whenever $x\neq y$, we get 
\[
\Delta^2=T^2-4h(1-x^a)(1-x^v)\ge 4\tilde{q}^a(1-x^a)\tilde{q}^v(1-x^v)-4h(1-x^a)(1-x^v)
\]
Now $\tilde{q}^a\tilde{q}^v=h^2+h(1-h)+(1-h)^2q(1-q)\ge h$, so we conclude $\Delta^2 > 0$. 

Since eigenvalues are all distinct, the matrix is diagonalizable, and it is negative definite whenever the eigenvalues are negative. Inspecting the expression, this happens whenever $\mu>\hat{\mu}$. \end{proof}

For the proof of Propositions \ref{prop:exogenoush} and \ref{prop:exogenoush2} we are going to need the following lemma. For convenience, given the extensions of Appendix \ref{sec:extensions}, we state the results for heterogeneous initial conditions $\rho_0^a,\rho_0^v$. The baseline case considered in the body of the paper is with $\rho_0^a=\rho_0^v=\rho_0$.

\begin{lemma}
\label{cuminf}
Let $(\rho_0^a,\rho_0^v)$ be the infected share for each group at the outbreak. Then in the linearized approximation around the (0,0) steady state:
\begin{align}
CI^a=&\frac{2 \left[\rho_0^a \left(\mu -(1-x^v)\tilde{q}^v\right)+\rho_0^v \left(1-x^a\right)
	\left(1-\tilde{q}^a\right)\right]}{(  T-2 \mu-\Delta) ( T-2 \mu+\Delta)};\\
CI^v=&\frac{2 \left[\text{$\rho $}_0^a \left(1-x^v\right) \left(1-\tilde{q}^v\right)+\text{$\rho $}_0^v \left(\mu -(1-x^a)\tilde{q}^a\right)\right]}{(  T-2 \mu-\Delta) ( T-2 \mu+\Delta)};\\
CI=&\frac{2\left[\rho_0^a\left(\mu +(1-x^v)(1-2\tilde{q}^v)\right)+\rho_0^v\left(\mu +(1-x^a)(1-2\tilde{q}^a)\right)\right]}{(  T-2 \mu-\Delta) ( T-2 \mu+\Delta)}.
\end{align}
\end{lemma}
\begin{proof}

The linearized dynamics is:

\begin{align*}
\dot{d\rho}(t)&=J d\rho(t)\\
d\rho(0)&=\rho_0 
\end{align*}
where $\rho_0=(\rho^a_0,\rho^v_0)$, that is:
\begin{align*}
\dot{d\rho}(t)&=M d\rho(0)\\
d\rho(0)&=\rho_0,\quad M=e^{tJ} 
\end{align*}
and:
\begin{align*}
M_{11}&=\frac{1}{\Delta }e^{\frac{1}{2} t (T-2 \mu )} \left(\sinh \left(\frac{\Delta  t}{2}\right) \left(-x^a \tilde{q}^a+\tilde{q}^a-\mu
+\frac{1}{2} (2 \mu -T)\right)+\frac{1}{2} \Delta  \cosh \left(\frac{\Delta  t}{2}\right)\right)\\
M_{12}&=\frac{1}{\Delta }\left(1-x^a\right) \left(1-\tilde{q}^a\right) \sinh \left(\frac{\Delta  t}{2}\right) e^{\frac{1}{2} t (T-2 \mu )}\\
M_{21}&=\frac{1}{\Delta }\left(1-x^v\right) \left(1-\tilde{q}^v\right) \sinh \left(\frac{\Delta  t}{2}\right) e^{\frac{1}{2} t (T-2 \mu )}\\
M_{22}&=\frac{1}{\Delta }e^{\frac{1}{2} t (T-2 \mu )} \left(\sinh \left(\frac{\Delta  t}{2}\right) \left(-x^v \tilde{q}^v+\tilde{q}^v-\mu +\frac{1}{2} (2
\mu -T)\right)+\frac{1}{2} \Delta  \cosh \left(\frac{\Delta  t}{2}\right)\right)
\end{align*}

The cumulative infection in time in the two groups can be calculated analytically by integration, since it is just a sum of exponential terms. Integration yield, for $CI^v$:
\begin{align*}
    CI^v&=\int_0^{\infty}d\rho^v(t)d t\\
    &=\frac{2 \left(\text{$\rho $}_0^a \left(1-x^v\right) \left(1-\tilde{q}^v\right)+\text{$\rho $}_0^v \left(\mu -(1-x^a)\tilde{q}^a\right)\right)}{(-\Delta -2 \mu +T) (\Delta -2 \mu +T)}+\\
    &\lim_{t->\infty} e^{\frac{1}{2} t (T-2 \mu )} \left(2 \Delta  \cosh \left(\frac{\Delta  t}{2}\right) \left(\rho _0^a \left(x_v-1\right)
   \left(\tilde{q}^v-1\right)+\rho _0^v \left((1-x_v)\tilde{q}^v+\mu -T\right)\right)+\right.\\
   &\left.\sinh \left(\frac{\Delta
    t}{2}\right) \left(\rho _0^v \left((T-2 \mu ) \left(2 \left(x_v-1\right) \tilde{q}^v+T\right)+\Delta ^2\right)-2 \rho _0^a
   (T-2 \mu ) \left(x_v-1\right) \left(\tilde{q}^v-1\right)\right)\right)
\end{align*}
and the limit is zero if $\mu>\hat{\mu}$ because the leading term is $Exp\left(\frac{1}{2} t (T-2 \mu )+\frac{\Delta}{2}\right)=\hat{\mu}-\mu$. An analogous reasoning for $CI^a$ yields:
\begin{align}
CI^a=\int_0^{\infty}d\rho^a(t)d t&=\frac{2 \left(\rho_0^a \left(\mu -(1-x^v)\tilde{q}^v\right)+\rho_0^v \left(1-x^a\right)
	\left(1-\tilde{q}^a\right)\right)}{(-\Delta -2 \mu +T) (\Delta -2 \mu +T)}\\
CI^v=\int_0^{\infty}d\rho^v(t)d t&=\frac{2 \left(\text{$\rho $}_0^a \left(1-x^v\right) \left(1-\tilde{q}^v\right)+\text{$\rho $}_0^v \left(\mu -(1-x^a)\tilde{q}^a\right)\right)}{(-\Delta -2 \mu +T) (\Delta -2 \mu +T)} 
\end{align}

The total CI in the population is $CI=q^aCI^a+(1-q^a)CI^v$
\[
CI=\frac{2}{(-\Delta -2 \mu +T) (\Delta -2 \mu +T)}\left(q^a \left(\rho_0^a \left(\mu -(1-x^v)\tilde{q}^v\right)+\rho_0^v \left(1-x^a\right)
\left(1-\tilde{q}^a\right)\right)+\right.
\]
\[\left.(1-q^a)\left(\text{$\rho_0 $}^a \left(1-x^v\right) \left(1-\tilde{q}^v\right)+\text{$\rho_0 $}^v \left(\mu -(1-x^a)\tilde{q}^a\right)\right)\right)
\]
\[=
\rho_0^a\frac{2\left(q^a  \left(\mu -(1-x^v)\tilde{q}^v\right)+(1-q^a)\left(1-x^v\right) \left(1-\tilde{q}^v\right)\right)}{(-\Delta -2 \mu +T) (\Delta -2 \mu +T)}+
\]
\[
\rho_0^v\frac{2(q^a\left(1-x^a\right)
	\left(1-\tilde{q}^a\right)+(1-q^a)\left(\mu -(1-x^a)\tilde{q}^a\right))}{(-\Delta -2 \mu +T) (\Delta -2 \mu +T)}
\]

If $\rho_0^a=\rho_0^v=\rho_0$ (the case considered in the main part of the paper):
\begin{align}
CI^a&=\rho_0\frac{2 \left( \left(\mu -(1-x^v)\tilde{q}^v\right)+ \left(1-x^a\right)
	\left(1-\tilde{q}^a\right)\right)}{(-\Delta -2 \mu +T) (\Delta -2 \mu +T)}\\
CI^v&=\rho_0\frac{2 \left( \left(1-x^v\right) \left(1-\tilde{q}^v\right)+ \left(\mu -(1-x^a)\tilde{q}^a\right)\right)}{(-\Delta -2 \mu +T) (\Delta -2 \mu +T)} \\
CI&=2\rho_0\frac{\mu-(1-x^a)(\tilde{q}^a-q)-(1-x^v)(\tilde{q}^v-1+q) }{(-\Delta -2 \mu +T) (\Delta -2 \mu +T)}
\end{align}

\end{proof}

\subsection*{Proofs for Propositions \ref{prop:exogenoush}, \ref{prop_whobetter} and \ref{prop:exogenoush2} }

\begin{proof}
We develop the calculations for generic $\rho_0^a$ and $\rho^v_0$, that is the more general case useful also for  Proposition  \ref{prop_whobetter} and \ref{prop:exogenoush2}.

First, note that
 $\mu>\hat{\mu}$ implies:
\begin{align*}
    \mu >& 1-x_a>h(1-x_a)\\
    \mu>&1-x_v>h(1-x_v)\\
    \mu>&\frac{h(1-x_a)}{1-(1-h)q}\\
    \mu>&\frac{h(1-x_v)}{1-hq}
\end{align*}

The expressions of the derivatives are:

\begin{align*}
\frac{\partial CI^a}{\partial h}&=\frac{(q-1) \left(x_a-1\right) \left(\mu +x_v-1\right) \left(\rho _0^a \left(\mu +x_v-1\right)-\rho _0^v \left(x_a+\mu
   -1\right)\right)}{2 \left(h \mu  \left(-q x_a+x_a+q x_v-1\right)+h \left(x_a-1\right) \left(x_v-1\right)+\mu  \left(q
   \left(x_a-x_v\right)+\mu +x_v-1\right)\right){}^2}\\
\frac{\partial CI^a}{\partial q^a}&=   \frac{(h-1) \left(x_a-1\right) \left(h \left(x_v-1\right)+\mu \right)
   \left(\rho _0^a \left(\mu +x_v-1\right)-\rho _0^v \left(x_a+\mu -1\right)\right)}{2 \left(h \mu  \left(-q x_a+x_a+q
   x_v-1\right)+h \left(x_a-1\right) \left(x_v-1\right)+\mu  \left(q \left(x_a-x_v\right)+\mu
   +x_v-1\right)\right){}^2}\\
\frac{\partial CI^a}{\partial x^a}&=   \frac{\left((h-1) q \left(x_v-1\right)+\mu +x_v-1\right) \left(\mu  (h (q-1)-q) \left(\rho _0^a-\rho
   _0^v\right)-h \rho _0^a \left(x_v-1\right)-\mu  \rho _0^v\right)}{2 \left(h \mu  \left(-q x_a+x_a+q x_v-1\right)+h
   \left(x_a-1\right) \left(x_v-1\right)+\mu  \left(q \left(x_a-x_v\right)+\mu +x_v-1\right)\right){}^2}\\
\frac{\partial CI^a}{\partial x^v}&=\frac{(h-1) (q-1)
   \left(x_a-1\right) \left(\mu  \left((h-1) q \left(\rho _0^v-\rho _0^a\right)+\rho _0^v\right)+h \left(x_a-1\right) \rho
   _0^v\right)}{2 \left(h \mu  \left(-q x_a+x_a+q x_v-1\right)+h \left(x_a-1\right) \left(x_v-1\right)+\mu  \left(q
   \left(x_a-x_v\right)+\mu +x_v-1\right)\right){}^2}
\end{align*}

\begin{align*}
\frac{\partial CI^v}{\partial h}&=
\frac{q \left(x_v-1\right) \left(x_a+\mu -1\right) \left(\rho _0^a \left(\mu +x_v-1\right)-\rho _0^v \left(x_a+\mu
   -1\right)\right)}{2 \left(h \mu  \left(-q x_a+x_a+q x_v-1\right)+h \left(x_a-1\right) \left(x_v-1\right)+\mu  \left(q
   \left(x_a-x_v\right)+\mu +x_v-1\right)\right){}^2}\\
\frac{\partial CI^v}{\partial q^a}&=   \frac{(h-1) \left(x_v-1\right) \left(h \left(x_a-1\right)+\mu \right)
   \left(\rho _0^a \left(\mu +x_v-1\right)-\rho _0^v \left(x_a+\mu -1\right)\right)}{2 \left(h \mu  \left(-q x_a+x_a+q
   x_v-1\right)+h \left(x_a-1\right) \left(x_v-1\right)+\mu  \left(q \left(x_a-x_v\right)+\mu
   +x_v-1\right)\right){}^2}\\
\frac{\partial CI^v}{\partial x_a}&=   -\frac{(h-1) q \left(x_v-1\right) \left(h \rho _0^a \left(\mu +\mu  (-q)+x_v-1\right)+\mu  q \rho
   _0^a+(h-1) \mu  (q-1) \rho _0^v\right)}{2 \left(h \mu  \left(-q x_a+x_a+q x_v-1\right)+h \left(x_a-1\right)
   \left(x_v-1\right)+\mu  \left(q \left(x_a-x_v\right)+\mu +x_v-1\right)\right){}^2}\\
   \frac{\partial CI^v}{\partial x_v}&=\frac{\left(h (q-1) \left(x_a-1\right)+q
   \left(-x_a\right)-\mu +q\right) \left(\mu  \left((h-1) q \left(\rho _0^v-\rho _0^a\right)+\rho _0^v\right)+h
   \left(x_a-1\right) \rho _0^v\right)}{2 \left(h \mu  \left(-q x_a+x_a+q x_v-1\right)+h \left(x_a-1\right)
   \left(x_v-1\right)+\mu  \left(q \left(x_a-x_v\right)+\mu +x_v-1\right)\right){}^2}
\end{align*}
and combining them, we get:
\begin{align*}
\frac{\partial CI}{\partial h}&=    \frac{\mu  (q-1) q \left(x_a-x_v\right) \left(\rho _0^a \left(\mu +x_v-1\right)-\rho _0^v \left(x_a+\mu
   -1\right)\right)}{2 \left(h \mu  \left(-q x_a+x_a+q x_v-1\right)+h \left(x_a-1\right) \left(x_v-1\right)+\mu  \left(q
   \left(x_a-x_v\right)+\mu +x_v-1\right)\right){}^2}\\
   \frac{\partial CI}{\partial q^a}&=   \frac{(h-1) \left(\rho _0^a \left(\mu +x_v-1\right)-\rho _0^v \left(x_a+\mu
   -1\right)\right) \left(h \left(x_a-1\right) \left(x_v-1\right)+\mu  \left(q \left(x_a-x_v\right)+x_v-1\right)\right)}{2
   \left(h \mu  \left(-q x_a+x_a+q x_v-1\right)+h \left(x_a-1\right) \left(x_v-1\right)+\mu  \left(q \left(x_a-x_v\right)+\mu
   +x_v-1\right)\right){}^2}\\
   \frac{\partial CI}{\partial x_a}&=   -\frac{q \left(h \left(x_v-1\right)+\mu \right) \left(h \rho _0^a \left(\mu +\mu 
   (-q)+x_v-1\right)+\mu  q \rho _0^a+(h-1) \mu  (q-1) \rho _0^v\right)}{2 \left(h \mu  \left(-q x_a+x_a+q x_v-1\right)+h
   \left(x_a-1\right) \left(x_v-1\right)+\mu  \left(q \left(x_a-x_v\right)+\mu +x_v-1\right)\right){}^2}\\
   \frac{\partial CI}{\partial x_v}&=   \frac{(q-1) \left(h
   \left(x_a-1\right)+\mu \right) \left(\mu  \left((h-1) q \left(\rho _0^v-\rho _0^a\right)+\rho _0^v\right)+h \left(x_a-1\right)
   \rho _0^v\right)}{2 \left(h \mu  \left(-q x_a+x_a+q x_v-1\right)+h \left(x_a-1\right) \left(x_v-1\right)+\mu  \left(q
   \left(x_a-x_v\right)+\mu +x_v-1\right)\right){}^2}
\end{align*}

Note that all the denominators are positive, so to control the sign from now on we focus on the numerators. In particular, 
if $\rho_0^a=\rho_0^v=\rho_0$, we can note that $CI$ is increasing in $h$ and $CI$ is increasing in $q$ if and only if $x^v>x^a$.

If initial conditions are symmetric:
\begin{align*}
\frac{\partial CI^a}{\partial h}>0 \Longleftrightarrow&    -(q-1) \rho _0^a \left(x_a-1\right) \left(x_a-x_v\right) \left(\mu +x_v-1\right)>0\\
\frac{\partial CI^a}{\partial q^a}>0 \Longleftrightarrow& -(h-1) \rho _0^a \left(x_a-1\right)
   \left(x_a-x_v\right) \left(h \left(x_v-1\right)+\mu \right)>0\\
\frac{\partial CI^a}{\partial x^a}>0 \Longleftrightarrow&  -\rho _0^a \left(h \left(x_v-1\right)+\mu \right)
   \left(\mu-(1-h)(1-q) \left(1-x_v\right)\right)>0\\
\frac{\partial CI^a}{\partial x^v}>0 \Longleftrightarrow&   (h-1) (q-1) \rho _0^a \left(x_a-1\right) \left(h \left(x_a-1\right)+\mu
   \right)>0
\end{align*}
Now, using the first four inequalities presented above, we can conclude that $\frac{\partial CI^a}{\partial h}>0$, $\frac{\partial CI^a}{\partial q^a}>0$, $\frac{\partial CI^a}{\partial x^a}<0$ and $\frac{\partial CI^a}{\partial x^v}<0$. 
Similarly, if $\rho^a_0=0$:
\begin{align*}
\frac{\partial CI^a}{\partial h}>0 \Longleftrightarrow&    -(q-1) \left(x_a-1\right) \rho _0^v \left(x_a+\mu -1\right) \left(\mu +x_v-1\right)>0\\
\frac{\partial CI^a}{\partial q^a}>0 \Longleftrightarrow&    -(h-1) \left(x_a-1\right) \rho _0^v
   \left(x_a+\mu -1\right) \left(h \left(x_v-1\right)+\mu \right)>0\\
\frac{\partial CI^a}{\partial x^a}>0 \Longleftrightarrow&   -(1-h)(1-q) \rho _0^v \left(\mu -(1-q)(1-h)(1-x_v)\right)>0\\
\frac{\partial CI^a}{\partial x^v}>0 \Longleftrightarrow&   (h-1) (q-1) \left(x_a-1\right)\rho _0^v \left(h \left(x_a-1\right) +\mu  \left((h-1) q
   +1\right)\right)>0
\end{align*}
and we conclude that $\frac{\partial CI^a}{\partial h}<0$, $\frac{\partial CI^a}{\partial q^a}<0$, $\frac{\partial CI^a}{\partial x^a}<0$ and $\frac{\partial CI^a}{\partial x^v}<0$. 

If $\rho^v_0=0$:
\begin{align*}
\frac{\partial CI^a}{\partial h}>0 \Longleftrightarrow& (q-1) \rho _0^a \left(x_a-1\right) \left(\mu +x_v-1\right){}^2>0\\
\frac{\partial CI^a}{\partial q^a}>0 \Longleftrightarrow&    (h-1) \rho _0^a \left(x_a-1\right) \left(\mu +x_v-1\right)
   \left(h \left(x_v-1\right)+\mu \right)>0\\
\frac{\partial CI^a}{\partial x^a}>0 \Longleftrightarrow&   -\rho _0^a\left(\mu -(1-q)(1-h)(1-x_v)\right) \left(h
   \left(\mu -(1-x_v)\right)+\mu (1-h) q\right)>0\\
\frac{\partial CI^a}{\partial x^v}>0 \Longleftrightarrow&   -(h-1)^2 \mu  (q-1) q \rho _0^a \left(x_a-1\right)>0
\end{align*}
and we conclude that $\frac{\partial CI^a}{\partial h}>0$, $\frac{\partial CI^a}{\partial q^a}>0$, $\frac{\partial CI^a}{\partial x^a}<0$ and $\frac{\partial CI^a}{\partial x^v}<0$. 

The other cases are analogous.
\end{proof}

\subsection*{Proof of Proposition \ref{convtime}}

\begin{proof}
From the proof of Proposition \ref{SS}, the eigenvalues are:
\begin{eqnarray}
e_1 & = &  \hat{\mu} - \mu    \nonumber \\
   e_2 & = & \hat{\mu} - \mu -\Delta . \nonumber 
\end{eqnarray}
Moreover, they are both decreasing in absolute value as $h$ increases (this is easy to see for $e_1$, given that $\hat{\mu}$ is positive and  increases in $h$, but it holds also for $e_2$). \end{proof}

\subsection*{Proof of Proposition \ref{vax-infection}}

\begin{proof}
Using equations \eqref{endovax} we obtain:
\[
\frac{\dd CI^v}{\dd h}=\frac{\partial CI^v}{\partial h}+\frac{\partial CI^v}{\partial x^v}\frac{\dd x^v}{\dd h}
\]
\[
=-\frac{k q \left(k \left(h^2 (q-1)+2 h (\mu -q)-\mu ^2+q\right)+\mu -1\right)}{2 \left(\mu ^2 ((h-1)
   k q+k+1)-\mu  ((h-1) h k q+h k+h+1)+(h-1) h k q+h\right)^2}
\]
that is positive if $k>\frac{\mu -1}{h^2 (-q)+h^2-2 h \mu +2 h q+\mu ^2-q}$. \end{proof}

\subsection*{Proof of Proposition \ref{existence}}

For the proof we need the following lemma, that characterizes the cultural substitution pattern of the socialization payoffs.

\begin{lemma}
\label{subs}
If the risk is estimated via the cumulative infection, and $d>\underline{d}$ ($(x^a)^*=0$) $\Delta U^a$ is decreasing in $q$, while $\Delta U^v$ is increasing in $q$ (\emph{cultural substitution}).

If the risk is proportional to non-vaccinated agents, and $d>\underline{d}$ ($(x^a)^*=0$), then $\Delta U^v$ is increasing in $q$, while there exist a $\overline{d}$ such that for $d>\overline{d}$ $\Delta U^a$ is increasing in $q$, while is decreasing if $d<\overline{d}$.

\end{lemma}

\begin{proof}

{\bf Risk as cumulative infection}

Let us consider the case in which $d>\overline{d}$, so that $(x^a)^*=0$. To differentiate the socialization payoffs, we need the derivative of $x^v$ with respect to $q$. Using the implicit function theorem we get:

\scalebox{0.8}{$
\frac{\dd x^v}{\dd q}=-\frac{\text{d$\rho $a} (h-1) (h-\mu ) x^v \left(x^v-1\right)}{h^2 \left(-\text{d$\rho $a}
   q+\text{d$\rho $a}+2 \left(\mu +\mu  (-q) x^v+x^v-1\right)^2\right)+\text{d$\rho $a} h (\mu 
   (q-2)+q)+\text{d$\rho $a} \mu  (\mu -q)+4 h \mu  \left(\mu -q x^v+x^v-1\right) \left(-\mu +(\mu 
   q-1) x^v+1\right)+2 \mu ^2 \left(\mu -q x^v+x^v-1\right)^2}
$}

\vs 
so that the total derivatives of the payoffs are:

\resizebox{1.3\linewidth}{!}{
$
\begin{aligned}
\frac{\dd \Delta U^v}{\dd q}   & = \frac{\partial \Delta U^v}{\partial q} +\frac{\partial \Delta U^v}{\partial x^v}  \frac{\dd x^v}{\dd q}\\
&=-\frac{(h-1) \rho _0 x^v \left(\rho _0 (h-\mu ) \left(\left(x^v-1\right) \left((h-1) q
   x^v-h\right)-\mu +\mu  x^v-1\right)+2 \mu  x^v \left(-h \mu +h (\mu  q-1) x^v+h+\mu  \left(\mu -q
   x^v+x^v-1\right)\right)\right)}{2 \left(-h \mu +h (\mu  q-1) x^v+h+\mu  \left(\mu -q
   x^v+x^v-1\right)\right) \left(h^2 \left(-q \rho _0+2 \left(\mu +\mu  (-q) x^v+x^v-1\right)^2+\rho
   _0\right)+h \rho _0 (\mu  (q-2)+q)+4 h \mu  \left(\mu -q x^v+x^v-1\right) \left(-\mu +(\mu  q-1)
   x^v+1\right)+\mu  \rho _0 (\mu -q)+2 \mu ^2 \left(\mu -q x^v+x^v-1\right)^2\right)}\\
\frac{\partial \Delta U^a}{\partial q}  &=\frac{\partial \Delta U^v}{\partial q} +\frac{\partial \Delta U^v}{\partial x^v}  \frac{\dd x^v}{\dd q}\\
&\frac{(h-1) \rho _0 x^v \left(\rho _0 (h-\mu ) \left(\left(x^v-1\right) \left((h-1) q
   x^v-h\right)-\mu +\mu  x^v-1\right)-2 \left(d (h-\mu ) \left(x^v-1\right)-\mu  x^v\right)
   \left(-h \mu +h (\mu  q-1) x^v+h+\mu  \left(\mu -q x^v+x^v-1\right)\right)\right)}{2 \left(-h \mu
   +h (\mu  q-1) x^v+h+\mu  \left(\mu -q x^v+x^v-1\right)\right) \left(h^2 \left(-q \rho _0+2
   \left(\mu +\mu  (-q) x^v+x^v-1\right)^2+\rho _0\right)+h \rho _0 (\mu  (q-2)+q)+4 h \mu 
   \left(\mu -q x^v+x^v-1\right) \left(-\mu +(\mu  q-1) x^v+1\right)+\mu  \rho _0 (\mu -q)+2 \mu ^2
   \left(\mu -q x^v+x^v-1\right)^2\right)}
\end{aligned}$}

\vs 
Under our assumptions the first expression is positive, the second is negative.

{\bf Proportional risk}

The equilibrium values in this case are:
\begin{align}
x^a&=0\ ,\\
x^v&=\frac{k}{(h-1) k q+k+1}\ ,
\end{align}
provided that $x^v<1$ and $x^a<d$, that are true respectively if: 
\begin{align}
k q(1-h)<1 \label{int1}\\
d>\underline{d}= \label{int2}    
\end{align}

First, let us focus on the case in which $d>\overline{d}$, so that $(x^a)^*=0$. In this case the socialization payoffs are in the main text. 

In case of the proportional infection risk plugging the expression above for $x^v$ into \eqref{socpayoff} and taking the derivatives we get:
\begin{align*}
\frac{\partial \Delta U^v}{\partial q}   & =\frac{(1-h) k^3 ((h-1) h k q+h k+h+1)}{((h-1) k q+k+1)^3}>0 \\
\frac{\partial \Delta U^a}{\partial q}  &= \frac{(h-1) k^2 (k-(d-h k) ((h-1) k q+k+1))}{((h-1) k q+k+1)^3} 
\end{align*}
The first expression is positive, thanks to the interiority condition $k q (1-h)<1$.

The second expression is positive if $d > \frac{h^2 k^2 q-h k^2 q+h k^2+h k+k}{h k q-k q+k+1}$. This is not redundant with the interiority condition \eqref{int2}: so for intermediate values of $d$ we get cultural substitution, for large values cultural complementarity, in particular if $d>\overline{d}=\frac{h^2 k^2 q-h k^2 q+h k^2+h k+k}{h k q-k q+k+1}$.





The (eventual) interior steady state is defined by: the equation $q^{\alpha}\Delta U^a-(1-q)^{\alpha}\Delta U^v=0$. Call $\Phi(q)=q^{\alpha}\Delta U^a-(1-q)^{\alpha}\Delta U^v$. We want to show that $\Phi$ has a zero in $(0,1)$. First, we show that in both specifications $\Delta U^a(q=0)>0$ and $\Delta U^v(q=1)<0$, so that $\lim_{q->0^+}\Phi(q)>0$, while $\lim_{q->1^-}\Phi(q)<0$. Hence, by the intermediate value theorem, there exist an interior steady state.

{\bf Proportional risk}
$\Delta U^a(q=0)=\frac{k (2 (k+1) (d-h k)-k)}{2 (k+1)^2}$, and is positive if $ d  >\frac{-2 h k^2-2 h k+2 k^2+k-2}{2 k+2}$, while $\Delta U^v(q=1)=\frac{k^2 (2 h ((h-1) k+k+1)+1)}{2 ((h-1) k+k+1)^2}$ is positive, thanks to the condition $k q (1-h)<1$ (eq \ref{int1}). 

{\bf Risk as cumulative infection}
\[
\Delta U^v(q=1)=\frac{\rho _0 \left(\rho _0 \left(h \left(x^v-1\right)+\mu -x^v\right)^2+4 x^v \left((\mu -1) (\mu
   -h)+h (\mu -1) x^v\right)\right)}{8 \left(-h \mu +h (\mu -1) x^v+h+(\mu -1) \mu \right)^2}
\]
and the numerator is positive, because in this case $\mu > \hat{\mu}$ implies $\mu>1$, while:
\[
\Delta U^a(q=0)=
\]
and is positive if $\rho_0$ is small enough and if:
\[ 
d>\frac{h \rho _0-\mu  \rho _0-4 x^{2 v}-4 \mu  x^v+4 x^v}{4 h \mu +4 h x^v-4 h-4 \mu ^2+4 \mu -4 \mu 
   x^v}
\]
Hence $d_q=\max\{\frac{h \rho _0-\mu  \rho _0-4 x^{2 v}-4 \mu  x^v+4 x^v}{4 h \mu +4 h x^v-4 h-4 \mu ^2+4 \mu -4 \mu 
   x^v}, \frac{-2 h k^2-2 h k+2 k^2+k-2}{2 k+2} \}$.

The steady state is unique and stable for $\alpha<0$ because the derivative of $\Phi$ is negative. To prove it, note that the expression is:
\begin{align*} 
&\frac{\dd}{\dd q}\left(q^{\alpha}\Delta U^a-(1-q)^{\alpha}\Delta U^v\right)=\\
&\alpha q^{\alpha-1}\Delta U^a+\alpha(1-q)^{\alpha-1}\Delta U^v+q^{\alpha}\frac{\dd}{\dd q}\Delta U^a-(1-q)^{\alpha}\frac{\dd\Delta U^v}{\dd q}<0
\end{align*}

In the case of cumulative infection it is negative thanks to Proposition \ref{subs}, and because $\alpha<0$. 

In the case of Proportional infection risk instead 
\[
q^{\alpha}\frac{\dd}{\dd q}\Delta U^a-(1-q)^{\alpha}\frac{\dd}{\dd q}\Delta U^v=
\]
\[
q^{\alpha}\left(\frac{\dd}{\dd q}\Delta U^a-\frac{(1-q)^{\alpha}}{q^{\alpha}}\frac{\dd}{\dd q}\Delta U^v\right)
\]
that in the steady state is:
\[
q^{\alpha}\left(\frac{\dd}{\dd q}\Delta U^a-\frac{\Delta U^a}{\Delta U^v}\frac{\dd}{\dd q}\Delta U^v\right)
\]
and plugging the expressions: 
\[
\frac{-d (1-h) k^2}{((h-1) k q+k+1)^2 (2 h ((h-1) k q+k+1)+1)}<0
\]
that thanks to the interiority condition $k q (1-h)<1$ we can see to be always negative. \end{proof}

\subsection*{Proof of Proposition \ref{q_decreasing}}

\begin{proof}
We calculate the derivative using the implicit function theorem. That is, we have to compute the derivatives of $q^{\alpha}\Delta U^a-(1-q)^{\alpha}\Delta U^v$:

We can evaluate the derivative using the implicit function theorem:
\[
\frac{\dd q}{\dd h}=-\frac{\frac{\dd}{\dd h}\left(q^{\alpha}\Delta U^a-(1-q)^{\alpha}\Delta U^v\right)}{\alpha q^{\alpha-1}\Delta U^a+\alpha (1-q^{\alpha-1})\Delta U^v-\frac{\dd}{\dd q}\left(q^{\alpha}\Delta U^a-(1-q)^{\alpha}\Delta U^v\right)}
\]

The denominator is negative thanks to Proposition \ref{existence}. The numerator is negative because from the expressions \ref{socpayoff} and Assumption we immediately get that $\Delta U^a$ is decreasing in $h$ and $\Delta U^v$ is increasing.




Hence $q$ is decreasing in $h$. \end{proof}

\subsection*{Proof of Proposition \ref{prop:alpha}}

\begin{proof}
The total derivative of $CI$ is:
\[
\frac{\dd CI}{\dd h}=\frac{\partial CI}{\partial h}+\frac{\partial CI}{\partial x^v}\left(\frac{\partial x^v}{\partial h} +\frac{\partial x^v}{\partial q}\frac{\dd q}{\dd h}  \right)+ \frac{\partial CI}{\partial q}\frac{\dd q}{\dd h}
\]

We prove that as $\alpha$ tends to 0, $\frac{\dd q}{\dd h}$ tends to 0 as well. This way, the derivative above is the same as in Proposition \ref{prop:endovax}, and is positive. 

Indeed, using the implicit function theorem:
\[
\frac{\dd q}{\dd h}=-\frac{q^{\alpha}\frac{\dd}{\dd h}\Delta U^a-(1-q)^{\alpha}\frac{\dd}{\dd h}\Delta U^v}{\alpha q^{\alpha-1}\Delta U^a+\alpha(1-q)^{\alpha-1}\Delta U^v+q^{\alpha}\frac{\dd}{\dd q}\Delta U^a-(1-q)^{\alpha}\frac{\dd}{\dd q}\Delta U^v}
\]
\[
=-\frac{q^{\alpha}}{q^{\alpha}}\frac{\frac{\dd}{\dd h}\Delta U^a-\frac{(1-q)^{\alpha}}{q^{\alpha}}\frac{\dd}{\dd h}\Delta U^v}{\alpha \frac{1}{q}\Delta U^a+\alpha\frac{1}{1-q}\frac{(1-q)^{\alpha}}{q^{\alpha}}\Delta U^v+\frac{\dd}{\dd q}\Delta U^a-\frac{(1-q)^{\alpha}}{q^{\alpha}}\frac{\dd}{\dd q}\Delta U^v}
\]
\[
-\frac{\frac{\dd}{\dd h}\Delta U^a-\frac{\Delta U^a}{\Delta U^v}\frac{\dd}{\dd h}\Delta U^v}{\alpha \frac{1}{q}\Delta U^a+\alpha\frac{1}{1-q}\frac{(1-q)^{\alpha}}{q^{\alpha}}\Delta U^v+\frac{\dd}{\dd q}\Delta U^a-\frac{\Delta U^a}{\Delta U^v}\frac{\dd}{\dd q}\Delta U^v}
\]

Now let $\alpha$ go to $-\infty$. To see what happens to $q$, let us analyze:
\[
\frac{\Delta U^a}{\Delta U^v}=\frac{(1-q)^{\alpha}}{q^{\alpha}}=\left(\frac{q}{(1-q)}\right)^{-\alpha}
\]
Now, as $-\alpha \to \infty $, unless $\frac{q}{(1-q)}\to 1$ the limit has to be either 0 or $\infty$, which is impossible because under the interiority conditions $\frac{\Delta U^a}{\Delta U^v}$ remains bounded. Hence as $\alpha \to -\infty$ we have $q\to \frac{1}{2}$.

Then We have that as $\alpha \to -\infty$ the denominator of the derivative goes to $-\infty$, and so $\frac{\dd q}{\dd h} \to 0^-$. \end{proof}

\subsection*{Proof of Proposition \ref{propCIend}} 

\begin{proof}
Using the derivatives computed in Proposition \ref{prop:exogenoush}, we find that the additional term due to the fact that vaccination rates adjust is:
\begin{align*}
&\frac{\partial CI }{\partial x^v}\frac{\dd x^v}{\dd h}+\frac{\partial CI }{\partial x^a}\frac{\dd x^a}{\dd h}=\\
&-\frac{\rho_0 q(1-q)}{2 (h \mu  (-q x^a+q x^v+x^v-1)+h (x^a-1) (x^v-1)+\mu  (\mu +q
   (x^a-x^v)+x^v-1))^2}\times\\
 & \frac{d k (1-q) q \left((\mu -h (1-x^v))^2-(\mu -h (1-x^a))^2\right)}{(h k+1)^2}
\end{align*}
which is negative because since $x^v>x^a$ we have:
\[
\left((\mu -h (1-x^v))^2-(\mu -h (1-x^a))^2\right)>0
\]

The total derivative instead is positive:
\begin{align*}
\frac{\dd CI }{\dd h}=&\frac{\partial CI }{\partial h}+\frac{\partial CI }{\partial x^v}\frac{\dd x^v}{\dd h}+\frac{\partial CI }{\partial x^a}\frac{\dd x^a}{\dd h}\\
=&\frac{\rho_0 q(1-q)(x^v-x^a)^2(\mu (1-h k)+h^2 k(2-x^a-x^v))}{2 (h \mu  (-q x^a+q x^v+x^v-1)+h (x^a-1) (x^v-1)+\mu  (\mu +q
   (x^a-x^v)+x^v-1))^2}>0
\end{align*}
\end{proof}

\subsection*{Proof of Proposition \ref{propCIend1}} 

\begin{proof}
In the case of an interior solution, the equilibrium is determined by:
\[
\begin{cases}
x^a=CI^a-d\\
x^v=CI^v
\end{cases}
\]

Using the implicit function theorem we get that the derivatives of the infection rates for $h=0$ are:
\begin{align*}
\frac{\dd x^a}{\dd h}&=\frac{\rho_0^2}{D K^2}\mu(1-q)(x^v-x^a)(\mu-(1-x^v))    \\
\frac{\dd x^v}{\dd h}&=   -\frac{\rho_0^2}{D K^2}\mu q(x^v-x^a)(\mu-(1-x^a))  
\end{align*}
where $K=(T-2\mu)^2-\Delta^2>0$ and $D=\mu (\mu-q(1-x^a)-(1-q)(1-x^v))>0$.

So, using the derivatives of $CI$ precedently computed, we obtain that the additional effect due to adjustment of vaccination rates is:
\begin{align*}
\frac{\partial CI }{\partial x^v}\frac{\dd x^v}{\dd h}+\frac{\partial CI }{\partial x^a}\frac{\dd x^a}{\dd h}&=-\frac{\rho_0^2}{D K^3}\mu^2(1-q)q(x^v-x^a)\left((1-x^v)(\mu-(1-x^a))-(1-x^a)(\mu-(1-x^v))\right)\\
&=\frac{\rho_0^2}{D K^3}\mu^2(1-q)q(x^v-x^a)^2>0
\end{align*}
so the additional effect is positive, hence $CI$ is more increasing than in the baseline case for small $h$. \end{proof}

\subsection*{Proof of Proposition \ref{prop:existence_appendix}}

\begin{proof}
Consider the function $\Phi(q)=q^{\alpha}\Delta U^a - (1-q)^{\alpha}\Delta U^v$. Both $\Delta U^a$ and $\Delta U^v$ are bounded from above and bounded away from 0, so when $q \to 0$ the negative term remains bounded while $q^{\alpha} \to \infty$ (because $\alpha<0$). The reverse happens when $q\to 1$. By the intermediate value theorem, there exist a solution $q^* \in (0,1)$.

Concerning stability, we can calculate the derivative of the function $F$:
\[
\frac{\dd \Phi}{\dd q}=\frac{d}{2 (k+1) (h k+1)^2}\times
\]
\[ \left(a q^{a-1} \left(d \left(-2 (h-1) h k^2 q+k+1\right)-2 h k (h k+1)\right)-2 d (h-1) h k^2 \left(q^a+(1-q)^a\right)\right.\]
\[\left.+a
	(1-q)^{a-1} \left(2 d (h-1) h k^2 q+d (k+1) (2 h k+1)+2 h k (h k+1)\right)\right)
\]
If $q\to 0$, $\frac{\dd \Phi}{\dd q}\to-\infty$, whereas if $q \to 1$ $\frac{\dd \Phi}{\dd q}\to +\infty$, so that, by continuity, there must be a stable steady state. If if $h \to 0$, $\frac{\dd F}{\dd q}\to\alpha d^2 2^{1-\alpha}<0$, so for $h$ in a neighborhood of 0 the steady state is unique and stable. \end{proof}

\subsection*{Proof of Proposition \ref{derivative_qh}}

\begin{proof}
For $h=0$ we have that $q=\frac{1}{2}$. We can compute the derivative using the implicit function theorem. The first derivative is the proof of Proposition \ref{prop:existence_appendix}. The second is
\[
\frac{\dd \Phi}{\dd h}=\frac{d}{2 (k+1) (h k+1)^2}\times
\]
\[
d k \left((1-q)^a (h k (d (-(k+2) q+k+1)-1)+d k q-1)\right.
\]\[\left.-q^a (d (k q (h (k+2)-1)+k+1)+h k+1)\right)\]
so that:
\[
\left. \frac{\dd q}{\dd h} \right|_{h=0}=-\frac{\frac{\dd \Phi}{\dd h}}{\frac{\dd \Phi}{\dd q}}=\frac{2 k + d k}{\alpha (2 d + 2 d k)}
\] 
and we can see that $q$ is always decreasing with homophily, but with a different level of intensity according to the magnitude of $\alpha$.

\subsection*{Proof of Proposition \ref{endogenousq}}

Using the implicit function theorem, we can analyze the behavior of cumulative infection for $h$ close to 0:
\[
\left. \frac{\dd CI}{\dd h} \right|_{h=0}=\frac{(k+1)}{4 (d k-2 (k+1) \mu +2)^2} \left(\frac{4 (d+2) k (d \rho_0^v k+(k+1) \mu  (\rho_0^a-\rho_0^v)-\rho_0^a+\rho_0^v)}{a d (k+1)}\right.
		\]
		\[\left.+\frac{d k \left(\rho_0^a \left(d k^2+2 (k+1) \mu -2\right)+\rho_0^v (d k
		(k+2)-2 (k+1) \mu +2)\right)}{\mu }\right)
\]

With a symmetric initial condition we get:
\[
\frac{\dd CI}{\dd h}_{\Big| h=0}=\frac{\rho_0^a k^2 \left(a d^2 (k+1)^2+2 (d+2) \mu \right)}{2 a \mu  (d k-2 (k+1) \mu +2)^2}
\]
which is positive if $\alpha <\frac{-2 d \mu -4 \mu }{d^2 k^2+2 d^2 k+d^2}$ and negative otherwise.\end{proof}

\subsection*{ Proof of Proposition \ref{betamu}}

\begin{proof}
The linearized dynamics is (from Lemma \ref{cuminf}):
\begin{align*}
\dot{\rho}^a&=\frac{1}{\Delta }e^{\frac{1}{2} t (T-2 \mu )} \left(\sinh \left(\frac{\Delta  t}{2}\right) \left(-x^a \tilde{q}^a+\tilde{q}^a
-\frac{1}{2} T\right)+\frac{1}{2} \Delta  \cosh \left(\frac{\Delta  t}{2}\right)\right)\rho_0^a\\
&+
\frac{1}{\Delta }\left(1-x^a\right) \left(1-\tilde{q}^a\right) \sinh \left(\frac{\Delta  t}{2}\right) e^{\frac{1}{2} t (T-2 \mu )}\rho_0^v\\
\dot{\rho}^v&=\frac{1}{\Delta }\left(1-x^v\right) \left(1-\tilde{q}^v\right) \sinh \left(\frac{\Delta  t}{2}\right) e^{\frac{1}{2} t (T-2 \mu )}\rho_0^a\\
&+\frac{1}{\Delta }e^{\frac{1}{2} t (T-2 \mu )} \left(\sinh \left(\frac{\Delta  t}{2}\right) \left(-x^v \tilde{q}^v+\tilde{q}^v -\frac{1}{2}  T\right)+\frac{1}{2} \Delta  \cosh \left(\frac{\Delta  t}{2}\right)\right)\rho_0^v
\end{align*}

In particular, it depends on $\mu$ just through the exponential term $e^{\frac{1}{2} t (T-2 \mu )}$. So we can rewrite it as:
\begin{align*}
\dot{\rho}^a&=e^{\frac{1}{2} t (T-2 \mu )} \mathcal{A}(t)\\
\dot{\rho}^v&=e^{\frac{1}{2} t (T-2 \mu )} \mathcal{V}(t)
\end{align*}
where $\mathcal{A}(t)$ and $ \mathcal{V}(t)$ do not depend on $\mu$. Now the discounted cumulative infection for anti--vaxxers is equal to:
\[
CI^a=\int_0^{\infty} e^{-\beta t}e^{\frac{1}{2} t (T-2 \mu )} \mathcal{A}(t) d t= \int_0^{\infty} e^{\frac{1}{2} t (T-2 (\mu+\beta ))} \mathcal{A}(t) d t
\]
which is precisely the expression for the non discounted cumulative infection in a model where the recovery rate is $\mu'=\mu+\beta$. \end{proof}


\bibliographystyle{chicago}
\bibliography{biblio}

\end{document}